\documentclass[final,3p,times,fleqn,sort&compress]{elsarticle}
\usepackage[pdf]{lfmath}
\usepackage{amsthm}
\usepackage{lineno}
\newtheorem{proposition}{Proposition}

\usepackage{times}
\usepackage{newtxmath}

\usepackage{pstricks}
\usepackage{pst-plot}

\usepackage{microtype}
\begin{document}
\allowdisplaybreaks
\begin{frontmatter}
\title{The role of the asymmetric Ekman dissipation term on the energetics of the two-layer quasi-geostrophic model}
\author[drlf]{Eleftherios Gkioulekas}
\ead{drlf@hushmail.com}
\address[drlf]{University of Texas Rio Grande Valley, School of Mathematical and Statistical Sciences, 1201 West University Drive, Edinburg, TX 78539-2909}

\begin{abstract}
In the two-layer quasi-geostrophic model, the friction between the flow at the lower layer and the surface boundary layer, placed beneath the lower layer, is modeled by the Ekman term, which is a linear dissipation term with respect to the horizontal velocity at the lower layer. The Ekman term appears in the governing equations asymmetrically; it is placed at the lower layer, but does not appear at the upper layer. A variation, proposed by Phillips and Salmon, uses extrapolation to place the Ekman term between the lower layer and the surface boundary  layer, or at the surface  boundary layer. We present theoretical results that show that in either the standard or the extrapolated configurations, the Ekman term dissipates energy at large scales, but does not dissipate potential enstrophy. It also creates an approximately symmetric stable  distribution of potential enstrophy between the two layers. The behavior of the Ekman term changes fundamentally at small scales. Under the standard formulation, the Ekman term will unconditionally dissipate energy and also dissipate, under very minor conditions, potential enstrophy at small scales. However, under the extrapolated formulation, there exist small ``negative regions'', which are defined over a two-dimensional phase space, capturing the distribution of energy per wavenumber between baroclinic energy and barotropic energy, and the  distribution of potential enstrophy per wavenumber between the upper layer and the lower layer, where the Ekman term may inject energy or potential enstrophy. 
\end{abstract}

\end{frontmatter}


\section{Introduction}

The two-layer quasi-geostrophic model is the most minimal vertical discretization of the quasi-geostrophic model, that captures the basic dynamics of atmospheric turbulence at planetary length scales (i.e. greater than $100$km) under the limits of rapid rotation and small vertical thickness. It consists of two vorticity-streamfunction equations, similar to two-dimensional Navier-Stokes, placed on an upper layer and a lower layer, at $0.25$ Atm and $0.75$ Atm respectively, and a temperature equation placed on a midlayer at $0.5$ Atm. The system is forced thermally, via the temperature equation, and dissipated by small-scale and large-scale dissipation terms placed on the vorticity equations. Using a potential vorticity reformulation, the temperature equation is eliminated and the two vorticity equations are replaced with two potential vorticity equations that are forced via anti-symmetric random forcing (see Appendix A of Ref.~\cite{article:Gkioulekas:p15} for details). Both potential vorticity equations have small-scale dissipation terms that model the dissipativity of the underlying three-dimensional dynamics at small scales. The potential vorticity equation corresponding to the lower layer also has a large scale dissipation term, known as the Ekman dissipation term, that models the dissipation effect resulting from friction of the flow with the surface boundary layer at 1Atm. This term is described as an asymmetric dissipation term because it is placed only on the potential vorticity equation for the lower layer, based on the modeling assumption that only the lower layer entertains friction with the surface boundary layer. 

In many ways, the two-layer quasi-geostrophic model dynamics is similar to that of the two-dimensional Navier-Stokes equations. In both models, the nonlinear interactions conserve energy and potential enstrophy, and, furthermore, the potential enstrophy of each layer is conserved separately. Based on a seminal paper by Charney \cite{article:Charney:1971}, the conventional wisdom has been, for some time, that the turbulence phenomenology of quasi-geostrophic model is isomorphic to that of the two-dimensional Navier-Stokes equations, consisting of an inverse energy cascade towards large scales and a downscale enstrophy cascade towards small scales, as predicted by the Kraichnan-Leith-Batchelor theory \cite{article:Kraichnan:1967:1,article:Leith:1968,article:Batchelor:1969} (hereafter KLB). This viewpoint was challenged in several subsequent papers \cite{article:Welch:2001,article:Orlando:2003,article:Orlando:2003:1,article:Tung:2007,article:Gkioulekas:p15}. The most important difference between the two-layer quasi-geostrophic model and two-dimensional Navier-Stokes is that the former does not retain the tight relationship $D_G (k) = k^2 D_G (k)$ between the energy dissipation rate spectrum $D_E (k)$ and the  enstrophy dissipation rate spectrum $D_G (k)$ \cite{article:Tung:2007}, which plays a fundamental role in establishing the direction of cascades in two-dimensional turbulence \cite{article:Tung:2007:1}. As a result, when the two layers of the two-dimensional quasi-geostrophic model are being dissipated asymmetrically, one cannot rule out the possibility of an observable downscale energy cascade. 

Tung and Orlando \cite{article:Orlando:2003} conducted a numerical simulation using a two-layer quasi-geostrophic model in which they observed coexisting downscale cascades of energy and potential enstrophy that resulted in a mixed energy spectrum exhibiting a transition from $k^{-3}$ scaling to $k^{-5/3}$ scaling with the transition wavenumber $k_t$ situated near the Rossby wavenumber $k_R$. Using dimensional analysis, Tung and Orlando \cite{article:Orlando:2003} argued that the transition wavenumber $k_t$ should depend on the downscale energy flux $\gee$ and the downscale enstrophy flux $\gn$ via the relation $k_t \sim \sqrt{\gn/\gee}$, and they have furthermore verified this relation as well as the downscale direction for both the energy flux and the potential enstrophy flux via simulation diagnostics. 

The phenomenology underlying the coexisting downscale energy and enstrophy cascades over the same inertial range can be understood in terms of a linear superposition principle, derived from the exact structure of the underlying statistical theory, that should hold for both two-dimensional Navier-Stokes turbulence and for the two-layer quasi-geostrophic model \cite{article:Tung:2005,article:Tung:2005:1}. According to this principle, each cascade contributes a power-law term to the energy spectrum $E(k)$, and the two terms combine linearly to give the total energy spectrum. In two-dimensional Navier-Stokes, a flux inequality limits the downscale energy flux severely, causing the contribution of the downscale enstrophy cascade to dominate over the entire downscale inertial range. However, this flux inequality does not necessarily persist in two-layer quasi-geostrophic models under asymmetric dissipation \cite{article:Tung:2007,article:Gkioulekas:p16}, and a violation of the flux inequality would correspond to a downscale energy flux strong enough to result in a broken energy spectrum with an observable transition from $k^{-3}$ scaling to $k^{-5/3}$ scaling with increasing wavenumbers $k$, where the energy cascade term overtakes the enstrophy cascade term after a transition wavenumber $k_t$ situated within the downscale inertial range. 

Tung and Orlando \cite{article:Orlando:2003} theorized that the observed Nastrom--Gage energy spectrum of the atmosphere \cite{article:Gage:1979,article:Nastrom:1986,article:Jasperson:1984,article:Gage:1984} results from coexisting downscale cascades of energy and potential enstrophy, and the point of their work was to demonstrate that such coexisting cascades can manifest even in a model as close to two-dimensional Navier-Stokes turbulence as the two-layer quasi-geostrophic model. Since then, their coexisting cascades theory has been corroborated by measurements and analysis  \cite{article:Zagar:2011} as well as by  numerical simulations of more realistic models  \cite{article:Lindborg:2011,article:Kurien:2011} that have also encountered coexisting downscale cascades. A recent numerical simulation of the two-layer quasigeostrophic model, under periodic boundary conditions and internal forcing, was not able to reproduce a violation of the flux inequality \cite{article:Watanabe:2019}. Due to the use of potential vorticity-based small-scale dissipation,  instead of streamfunction-based dissipation,  we are not able to use our previous results \cite{article:Gkioulekas:p16} to ascertain whether the asymmetry in the dissipation terms between the two layers was sufficiently strong; a question that will be investigated in future work.

Beyond the controversies relating to understanding the Nastrom-Gage spectrum \cite{article:Tung:2006,article:Gkioulekas:p15}, downscale energy cascades in mathematical models, such as the two-layer quasi-geostrophic model, are very intriguing from the point of view of fundamental turbulence research; the model itself is simple enough that its investigation may be possible using techniques that have been successful with two-dimensional Navier-Stokes turbulence \cite{article:Lebedev:1994,article:Lebedev:1994:1,article:Yakhot:1999,article:Procaccia:2002,article:Gkioulekas:2008:1,article:Gkioulekas:p14}. More importantly, there is the open problem of explaining why the downscale energy cascade in three-dimensional Navier-Stokes turbulence has intermittency corrections whereas the inverse energy cascade in two-dimensional Navier-Stokes turbulence follows intermittency-free Kolmogorov scaling \cite{article:Vergassola:2000,article:Tabeling:2002}, where further insight may be gained if one ever studies intermittency in downscale energy cascades manifesting in the two-layer or multi-layer quasi-geostrophic models. 

In order for an observable downscale energy cascade to manifest itself under two-layer quasi-geostrophic turbulence we need the confluence of two requirements: First, the ratio of the rate $\gn$ of injected potential enstrophy over the rate $\gee$ of injected energy from the forcing range and into the downscale inertial range, accounting for any energy and potential enstrophy dissipation at the forcing range itself, needs to place the transition wavenumber $k_t \sim \sqrt{\gn/\gee}$ within the downscale inertial range, in order to have enough downscale energy flux to generate an observable downscale energy cascade. Second, the flux inequality, mentioned previously, should be violated at large wavenumbers in order to ensure that the increased downscale energy flux can be dissipated. Both requirements were investigated rigorously in previous papers \cite{article:Gkioulekas:p15,article:Gkioulekas:p16} and both investigations have been inconclusive, or at best speculative, because they were grounded in rigorous mathematics, avoiding phenomenological assumptions about  two-layer quasi-geostrophic turbulence. 

In Ref.~\cite{article:Gkioulekas:p15} we showed that under random thermal forcing the injection ratio $\gn/\gee$ will place the transition wavenumber $k_t$ near the Rossby wavenumber $k_R$, if all of the injected energy and potential enstrophy cascades towards the small scales. It remains unclear how the Ekman dissipation term modifies this result. We speculated that if the asymmetric Ekman term suppresses random forcing on the lower layer, then the transition wavenumber $k_t$ would be decreased Although the inference itself is rigorous, it is not clear whether that is the effect that the Ekman term really has on the random forcing at the forcing range. 

In Ref.~\cite{article:Gkioulekas:p16}  we have studied the flux inequality under a very wide range of dissipation term configurations. We have shown that it is possible to rigorously prove negative  results that state that if the asymmetry between the dissipation terms, placed at the upper and lower layers,  is less than some upper bound, then the flux inequality will not be violated. Such results can be derived without making any phenomenological assumptions about the behavior of the two-layer quasi-geostrophic model. Unfortunately, the results that we would like to have, namely sufficient conditions on the dissipation asymmetry for violating the flux inequality, cannot be obtained without some knowledge of the underlying phenomenology. We have offered some speculations about some dissipation term configurations facilitating a flux inequality violation more effectively than others, but this question also remains open. 

In this paper we report on some new results towards resolving the first question. Our focus is to study the energy and potential enstrophy dissipation rate spectra $D_E (k)$ and $D_G (k)$ of the asymmetric Ekman term and draw out some phenomenological insight about what it does to the downscale injection rates of energy and potential enstrophy, and how it affects the overall dynamics of the two-layer quasi-geostrophic model. The main breakthrough that allows us to make progress is the following idea: we separate the energy spectrum $E(k)$ into a barotropic energy spectrum $E_K (k)$ and a baroclinic energy spectrum $E_P (k)$, such that $E(k) = E_K (k)+E_P (k)$, and we assume to be given the function $P(k)$ controlling the distribution of energy between barotropic and baroclinic at the wavenumber $k$, such that $E_K (k) = [1-P(k)] E(k)$ and $E_P (k) = P(k) E(k)$. We also separate the potential enstrophy spectrum $G(k)$ into the potential enstrophy spectrum $G_1 (k)$ of the upper layer and the potential enstrophy spectrum $G_2 (k)$ of the lower layer, such that $G(k) = G_1 (k) + G_2 (k)$, and we assume to be given the function $\gC (k)$ controlling the distribution of potential enstrophy between the two layers, such that $G_1 (k) = \gC (k) G(k)$ and $G_2 (k) = [1-\gC (k)] G(k)$. It is then possible to calculate the energy and potential enstrophy dissipation rate spectra $D_E (k)$ and $D_{G_2} (k)$ in terms of $\gC (k)$, $P(k)$, and $E(k)$, and hope that assumptions about $\gC (k)$ and $P(k)$ can tell us something interesting about the Ekman term.  

To account for the function $P(k)$, governing the distribution of energy between baroclinic and barotropic  energy,  we can rely on the phenomenology  that was proposed by Salmon \cite{article:Salmon:1978,article:Salmon:1980}, according to which, energy is initially injected as baroclinic energy at the forcing scale, and cascades downscale to the Rossby  wavenumber $k_R$ length scale, where some of the baroclinic  energy is converted to barotropic  energy. Subsequently, starting at the Rossby  wavenumber $k_R$, most of the resulting barotropic energy cascades upscale,  and is dissipated at large scales by the Ekman dissipation term,  and some cascades downscale, and is dissipated  by the small-scale dissipation term. Furthermore, a  portion of the baroclinic energy that is not converted at the Rossby  wavenumber,  continues to cascade downscale, and is also dissipated by  small-scale dissipation.  At steady state,  all this adds up  to a weak downscale cascade of the total energy.  In the limit $k\ll k_R$, most of the energy is initially baroclinic, since it is injected in that way.  However, after it converts to barotropic  and cascades upscale, the energy distribution becomes predominantly barotropic.  In the limit $k\gg k_R$, there  is a comparable  amount of both barotropic  and baroclinic  energy cascading  in the downscale direction. This implies that, in the limit $k\ll k_R$, we expect that $P(k)$ is initially near $1$ but converges  to a value near $0$, whereas, in the limit $k\gg k_R$, $P(k)$ converges  to a value that is placed near the midpoint between $0$ and $1$.

The asymmetric Ekman term is essential to facilitating the Salmon phenomenology, as it is needed to dissipate the upscale barotropic energy cascade, and, in fact, an early study by Holopainen \cite{article:Holopainen:1961} first hinted that the Ekman term triggers the conversion of energy from barotropic to baroclinic. Salmon justified his predicted phenomenology by a triad interactions argument and confirmed it via an EDM closure model numerical simulation \cite{article:Salmon:1978}  in which symmetric Ekman terms were placed on both layers. In two out of three simulations, the energy spectrum was predominantly barotropic, although increasing the Ekman term coefficient $\nu_E$ resulted in a predominantly baroclinic energy spectrum at the forcing range. Similar results were obtained by Rhines \cite{article:Rhines:1979}, who considered a freely decaying two-layer model,  and were confirmed by an EDQNM closure model simulation \cite{article:Sadourny:1982}, and later on by direct numerical simulation \cite{article:Held:1995}, in which an asymmetric Ekman term was used.  Subsequent work \cite{article:Flierl:2003,article:Flierl:2004,article:Vallis:2014}  has shown that the Salmon phenomenology is altered,  as the coefficient of the asymmetric Ekman term is increased,  in the $k \ll k_R$ limit as follows: there is an intermediate regime in which an increased asymmetric Ekman term inhibits the converted barotropic energy from forming an inverse barotropic energy cascade, which tends to shift the energy distribution towards baroclinic. This regime is then followed with a more extreme regime in which the flow on the lower layer is suppressed and the energy distribution becomes predominantly baroclinic.  

Accounting for the function $\gC (k)$,  governing the distribution of potential enstrophy between the two layers is simpler.  Due to  the layer-by-layer conservation of potential enstrophy, it cannot redistribute itself between the two layers via nonlinear interactions. Consequently, the potential enstrophy distribution $\gC (k)$ is affected solely by the forcing and dissipation terms, as  will be further discussed in the present paper.

Because the details of the results presented in the paper are very technical, we shall now  provide a detailed informal account of the predicted phenomenology in the rest of this introductory section. First of all, we have found that $\gC (k)$ and $P(k)$ are not entirely independent of each other, but are restricted by a mathematically rigorous inequality. In physical terms, the inequality implies that when almost all of the energy is initially injected as baroclinic energy in the limit $k\ll k_R$, i.e. $P(k) \approx 1$, then $\gC (k)$ is restricted to a very tight interval around $1/2$, corresponding to equal distribution of potential enstrophy between the upper layer and the lower layer. As the distribution of energy shifts from baroclinic to barotropic, the restriction on $\gC (k)$ widens allowing a greater percentage of potential enstrophy to concentrate on one layer versus the other. This is relevant because random thermal forcing, which corresponds to anti-symmetric random forcing of the potential vorticity equation \cite{article:Gkioulekas:p15}, injects only baroclinic energy at the forcing range. We may, therefore, expect that throughout the forcing range we initially have $P(k) \approx 1$ and therefore equal distribution of potential enstrophy between the two layers, i.e. $\gC (k) \approx 1/2$ or equivalently $G_1 (k) \approx G_2 (k)$. 

Our first major mathematical result is that the asymmetric Ekman term actually tends to stabilize the approximate equipartition of potential enstrophy that is initially caused by the exclusively baroclinic energy injection. More precisely, we show that for $k\ll k_R$ when $G_1 (k) = G_2 (k)$, the asymmetric Ekman term \emph{removes} potential enstrophy from the lower layer, thereby increasing the ratio $G_1 (k)/G_2 (k)$. Before that ratio has a chance to increase much, the asymmetric Ekman term now becomes injective and adds potential enstrophy to the lower layer, decreasing the ratio $G_1 (k)/G_2 (k)$ at the forcing range. Consequently, at steady state we expect the ratio $G_1 (k)/G_2 (k)$ to settle down on a stable fixed point where no potential enstrophy is being dissipated at the forcing range. The location of the fixed point will vary  as a function of the wavenumber ratio $k/k_R$ but it will maintain an approximate equipartition of potential enstrophy between the two layers for all wavenumbers $k\ll k_R$, with more potential enstrophy concentrated in the upper layer.

Our next major result is that the dynamic behavior of the asymmetric Ekman term changes in the limit $k\gg k_R$, where it becomes exclusively dissipative with respect to potential enstrophy. More precisely, for wavenumbers $k\gg k_R$, the Ekman term will dissipate potential enstrophy only from the lower layer, but not from the upper layer. Furthermore, since potential enstrophy is being conserved separately for each layer by the nonlinear interactions, it cannot be redistributed between layers by the nonlinear interactions. We expect therefore that the ratio $G_1 (k)/G_2 (k)$ will increase with increasing wavenumbers $k$ in the limit $k\gg k_R$, provided that the Ekman term coefficient is sufficiently large to sustain the potential enstrophy dissipation rate spectrum $D_G (k)$ at these wavenumbers. A sufficient increase in the ratio $G_1 (k)/G_2 (k)$ may facilitate the violation of the flux inequality, as was first noted in Ref.~\cite{article:Tung:2007}. However, a stronger Ekman term does not dissipate potential enstrophy for the wavenumbers $k\ll k_R$ where we expect to see the energy spectrum scaling $k^{-3}$ of the downscale potential enstrophy cascade dominate, because of the stable fixed-point partition of potential enstrophy between the upper and lower layer. Consequently, we do not expect the Ekman term to disrupt the $k^{-3}$ part of the broken energy spectrum.

Although the Ekman term's behavior is ambivalent with respect to potential enstrophy, where it may inject or dissipate potential enstrophy, depending on the distribution of potential enstrophy between the two layers, we show that, in its standard form, it is always dissipative with respect to energy for all wavenumbers. As a result, the overall picture is that both potential enstrophy and energy are injected at the forcing range with injection ratio $\gn/\gee \sim k_R^2$. The Ekman term dissipates some of the injected energy but does not dissipate the injected potential enstrophy at the injection wavenumbers, so the resulting downscale fluxes of energy and potential enstrophy shift the transition wavenumber $k_t$ towards small scales, i.e. $k_t > k_R$. This is the opposite of what we would have expected to see from assuming that the Ekman term merely dampens the forcing term at the lower layer \cite{article:Gkioulekas:p15}, indicating that such an assumption is an oversimplification. Furthermore, we see some tension between two opposing tendencies: a strong Ekman term is needed to violate the flux inequality and result in placing the transition wavenumber $k_t$ in the inertial range. On the other hand, when the Ekman term is too strong, it may end up dissipating too much energy at the forcing range, resulting in an insufficient amount of downscale energy flux, thereby pushing the transition wavenumber $k_t$ back into the dissipation range. 

Finally, in this paper we will also consider the behavior of a modified form of the asymmetric Ekman term that we have previously described as \emph{extrapolated Ekman dissipation} \cite{article:Gkioulekas:p16}. The standard formulation of the Ekman term makes it dependent only on the streamfunction of the lower layer.  In the extrapolated formulation, which was initially proposed by Phillips  \cite{article:Phillips:1956} and Salmon \cite{article:Salmon:1980}, the Ekman term is dependent on the streamfunctions of both of the  upper and the lower layer, appearing again only on the potential vorticity equation of the lower layer. The rationale for the extrapolated formulation is that the Ekman term depends on the streamfunction field at the surface boundary layer at 1 Atm, situated below the lower layer, which is typically placed at 0.75 Atm. In the extrapolated formulation, the streamfunction at the surface layer is modeled via linear extrapolation from the streamfunction at the upper and lower layers, whereas in the standard formulation the surface layer streamfunction is set equal to the lower layer streamfunction. Our discussion, so far, detailed what happens when the standard form of the Ekman term is used asymmetrically at the lower layer. So, what changes if we instead use the extrapolated formulation of the Ekman term? 

First, we have found that for wavenumbers $k\ll k_R$ our previous argument regarding the potential enstrophy dissipation continues to hold. The potential enstrophy distribution is stabilized in an approximate equipartition between the upper and lower layers, and as a result no potential enstrophy should be dissipated on average at steady state when $k\ll k_R$. For wavenumbers $k\gg k_R$, the extrapolated Ekman term will dissipate potential enstrophy only from the lower layer, thereby increasing the ratio $G_1 (k)/G_2 (k)$, which we expect to help break the flux inequality at large wavenumbers. As a result, the behavior of the extrapolated Ekman term with respect to potential enstrophy dissipation is not different from that of the standard Ekman term. 

However, there are differences with respect to energy dissipation. For  wavenumbers in the limit $k \ll k_R$, where the flow is expected to be predominantly barotropic, according to Salmon's phenomenology \cite{article:Salmon:1978,article:Salmon:1980,book:Salmon:1998}, the extrapolated Ekman term will dissipate energy, regardless of the distribution of potential enstrophy between the two layers. In the limit of very high Ekman dissipation coefficient $\nu_E$, the flow may shift towards being more baroclinic, and when at least half of the energy in the energy spectrum is baroclinic,  there are negative regions where the extrapolated Ekman dissipation term will inject energy instead of dissipating it. These regions occur when the flow is at least half baroclinic and most potential enstrophy is concentrated in the upper layer.  In the limit $k\gg k_R$, the energy in the energy spectrum is expected to be distributed in comparable amount between barotropic energy and baroclinic energy.  The extrapolated Ekman term will be dissipative if at least one half of the energy in the energy spectrum is baroclinic. When more than half of the energy  is instead barotropic, then there is a negative region where the extrapolated  Ekman term injects energy,  when most of the potential enstrophy  is also concentrated in the upper layer. More details about this strange behavior of the energy dissipation rate spectrum, under the extrapolated Ekman term, is given in the discussion of the negative regions displayed in Fig.~\ref{fig:BE-one} and Fig.~\ref{fig:BE-two}. 

This paper is organized, as explained in the following. Section 2 describes the governing equations of the two-layer quasi-geostrophic model and gives the mathematical definition of the bracket notation, used to define spectra of energy and potential enstrophy as well as the corresponding dissipation rate spectra. Section 3  defines the energy and potential enstrophy spectrum functions $E(k)$, $G_1 (k)$, $G_2 (k)$, the baroclinic energy spectrum $E_P(k)$, the barotropic energy spectrum $E_K(k)$, and the streamfunction spectra $U_1 (k)$, $U_2 (k)$, $C_{12}(k)$, and shows how all of them can be calculated in terms of $E(k)$ and the functions $P(k)$ and $\gC(k)$. We also derive Proposition~\ref{prop:distribution-constraint}, establishing a rigorous mathematical constraint between $P(k)$ and $\gC(k)$. Section 4 writes the energy and potential enstrophy dissipation rate spectra $D_E (k)$, $D_{G_1} (k)$, $D_{G_2}(k)$ in terms of $E(k)$, $P(k)$, and $\gC(k)$. Section 5 studies the predicted phenomenology of the potential enstrophy dissipation rate spectrum $D_{G_2}(k)$ in the limits $k \ll k_R$ and $k \gg k_R$. Section 6 presents a similar study for the energy dissipation rate spectrum $D_E(k)$. The paper concludes with Section 7. Several technical details and proofs are given in \ref{sec:app-Ekman-term}, \ref{sec:sign-of-BG-two}, \ref{sec:sign-of-BG-one}, and \ref{app:fig-two-geometry}. 

\section{Preliminaries}

The potential vorticity formulation of the two-layer quasi-geostrophic model is given by the following two governing equations for the potential vorticity in each layer:
\begin{align}
&\pderiv{q_1}{t}+J(\gy_1, q_1) = d_1 + f_1, \\
&\pderiv{q_2}{t}+J(\gy_2, q_2) = d_2 + f_2.
\end{align}
Here, $q_1,\gy_1$ represent the potential vorticity and the streamfunction at the upper layer placed at $0.25$ Atm; $q_2,\gy_2$ represent the potential vorticity and the streamfunction at the lower layer, placed st $0.75$ Atm; $d_1, d_2$ represent the dissipation terms at the upper and lower layer; $f_1, f_2$ represent the random forcing terms at the upper and lower layer. The nonlinear terms are represented by $J(\gy_1, q_1)$ and $J(\gy_2, q_2)$, where the general definition reads:
\begin{equation}
J(a,b) = \frac{\pd(a,b)}{\pd(x,y)} = \pderiv{a}{x}\pderiv{b}{y}-\pderiv{b}{x}\pderiv{a}{y}.
\end{equation}
The potential vorticities $q_1,q_2$ are related with the streamfunctions $\gy_1,\gy_2$ via
\begin{align}
q_1 &= \del^2\gy_1+f+\frac{k_R}{2}(\gy_2-\gy_1), \\
q_2 &= \del^2\gy_2+f-\frac{k_R}{2}(\gy_2-\gy_1),
\end{align}
with $k_R$ representing the Rossby wavenumber and $f=f_0+\gb y$ (with $f_0, \gb$ constants) representing the Coriolis term. Under the approximation $\gb =0$, $f$ becomes constant and is completely eliminated from the nonlinear terms. As we noted in a previous paper \cite{article:Gkioulekas:p16} this assumption is appropriate for the case of the Earth. The baroclinic instability is accounted for by the random forcing terms $f_1, f_2$, which must be defined antisymmetrically (i.e. $f_1=\varphi$ and $f_2=-\varphi$), under the assumption that all forcing is thermal \cite{article:Gkioulekas:p15}. 

For the dissipation terms $d_1, d_2$ we have previously \cite{article:Gkioulekas:p16} considered a broad range of several possible configurations, all encompassed by the equations:
\begin{align}
d_1 &= \nu (-1)^{p+1}\del^{2p+2}\gy_1, \\
d_2 &= (\nu+\gD\nu)(-1)^{p+1}\del^{2p+2}\gy_2-\nu_E\del^2\gy_s.
\end{align}
Here $\nu$ and $\nu+\gD\nu$ are the hyperviscosity coefficients for the small scale dissipation placed at both layers; $\nu_E$ is the coefficient of the Ekman dissipation term, that appears asymmetrically only on the lower layer; $\gy_s$ is the surface layer streamfunction, with the surface layer positioned anywhere between the lower layer at $0.75$ Atm and the surface boundary layer at $1$ Atm.

The standard choice for the Ekman term is to let $\gy_s =\gy_2$, corresponding to what we shall call \emph{standard Ekman term}. Another possibility \cite{article:Salmon:1980,article:Gkioulekas:p16} is to use linear extrapolation to express $\gy_s$ in terms of $\gy_1,\gy_2$. If $p_1$ is the pressure at the upper layer, $p_2$ is the pressure at the lower layer, and $p_s$ is the pressure at the surface layer, we require that the points with coordinates $(p_s,\gy_s), (p_1,\gy_1),(p_2,\gy_2)$ be collinear, as a means of extrapolating $\gy_s$ from $\gy_1,\gy_2$. It follows that $\gy_s =\gl\gy_2+\mu\gl\gy_1$ with $\gl=(p_s-p_1)/(p_2-p_1)$ and $\mu=(p_2-p_s)/(p_s-p_1)$ and we can furthermore show that $\gl=1/(\mu+1)$ and rewrite the equation for $\gy_s$ as
\begin{equation}
\gy_s = \frac{\mu\gy_1+\gy_2}{\mu+1}.
\end{equation}
For the most general case $0<p_1<p_2\leq p_s$ corresponding to stacking up the upper, lower, and surface layers in the right order, we can show that $-1<\mu\leq 0$. However, if we set $p_1=0.25$ Atm and $p_2=0.75$ Atm and assume that $p_2\leq p_s\leq 1$ Atm, then the range for $\mu$ narrows down to $-1/3\leq\mu\leq 0$. The case $\mu =0$ corresponds to the standard Ekman term, and the case $-1/3\leq\mu <0$ corresponds to the extrapolated Ekman term, with $\mu=-1/3$ corresponding to the extrapolated Ekman term formulation used by Phillips \cite{article:Phillips:1956} where $p_s = 1$. The details are given in \ref{sec:app-Ekman-term}. Understanding the effect of the Ekman term on the phenomenology of the two-layer quasi-geostrophic model is the main focus of this paper. 

The two-layer quasi-geostrophic model conserves energy as well as potential enstrophy in the upper layer and potential enstrophy in the lower layer. In Ref.~\cite{article:Gkioulekas:p15,article:Gkioulekas:p16} we have used the following bracket notation to define the energy spectrum $E(k)$ and the potential enstrophy spectra $G_1(k)$ and $G_2(k)$ for the upper and lower layers. Let $a(\bfx)$ with $\bfz\in\bbR^2$ be some field in $\bbR$ and let $k\in (0,+\infty)$. We define the filtered field $a^{<k}(\bfx)$ via the equation 
\begin{equation}
a^{<k}(\bfx) = \int_{\bbR^2}\df{\bfx_0}\int_{\bbR^2}\df{\bfk_0} \frac{H(k-\nrm{\bfk_0})}{4\pi^2}\exp (i\bfk_0\cdot (\bfx-\bfx_0)) a(\bfx_0),
\end{equation}
with $H(x)$ the Heaviside function defined by
\begin{equation}
H(x) = \casethree{1}{x\in (0,+\infty)}{1/2}{x=0}{0}{x\in (-\infty, 0).}
\end{equation}
This is a low-pass filter where $a^{<k}(\bfx)$ retains only the Fourier modes inside a disk in Fourier space with radius less than $k$, setting all modes outside of the disk equal to zero. Given two fields $a(\bfx)$ and $b(\bfx)$ with $\bfx\in\bbR$, with Fourier transforms $\hat{a}(\bfk)$ and $\hat{b}(\bfk)$ such that 
\begin{align}
a(\bfx) &= \int_{\bbR^2}\hat{a}(\bfk)\exp(i\bfk\cdot\bfx)\;\df{\bfk}, \\
b(\bfx) &= \int_{\bbR^2}\hat{b}(\bfk)\exp(i\bfk\cdot\bfx)\;\df{\bfk},
\end{align}
and given a wavenumber $k\in (0,+\infty)$, we define the bracket $\innerf{a}{b}{k}$ such that 
\begin{equation}
\innerf{a}{b}{k} = \dD{k}\int_{\bbR^2}\df{\bfx}\; \avg{a^{<k}(\bfx) b^{<k}(\bfx) }
= \frac{1}{2}\int_{A\in\SO (2)}\df{\gW (A)}\; k \avg{\hat{a}(kA\bfe) \hat{b}^{\ast}(kA\bfe) + \hat{a}^{\ast}(kA\bfe) \hat{b}(kA\bfe)}.
\label{eq:bracket-definition}
\end{equation}
Here, $\avg{\cdot}$ represents an ensemble average, $\SO{2}$ is the set of all non-reflecting rotation matrices in $\bbR^2$, $\df{\Omega (A)}$ represents the measure of the corresponding spherical integral over all rotations in $\bbR^2$, and $\bfe$ is a two-dimensional unit vector pointing in some arbitrarily chosen direction. The star notation in $\hat{a}^{\ast}$ and $\hat{b}^{\ast}$ represents a complex conjugate. 

It immediately follows that the bracket is symmetric and bilinear in that it satisfies, for all $\gl, \mu\in\bbR$
\begin{align}
&\innerf{a}{b}{k} = \innerf{b}{a}{k}, \\
&\innerf{a}{\gl a+\mu b}{k} = \gl\innerf{a}{b}{k} + \mu\innerf{a}{c}{k}, \\
&\innerf{\gl a + \mu b}{c}{k} = \gl\innerf{a}{c}{k} + \mu\innerf{b}{c}{k}.
\end{align}
We can also show that for any field $a(\bfx)$, the bracket is positive definite:
\begin{equation}
\innerf{a}{a}{k} \geq 0.
\end{equation}
Finally, we can show, as an immediate consequence of Eq.~\eqref{eq:bracket-definition} that
\begin{equation}
\innerf{\lapl a}{b}{k} = \innerf{a}{\lapl b}{k} = -k^2 \innerf{a}{b}{k}.
\end{equation}

\section{Energy and potential enstrophy spectra}
\label{sec:energy-potential-enstrophy-spectra}

The nonlinear terms of the two-layer quasi-geostrophic model conserve the total energy $E$ and the total potential enstrophies $G_1$ and $G_2$ for the upper and lower layers, given by
\begin{align}
E(t) &= -\int_{\bbR^2}\df{\bfx}\; [\gy_1 (\bfx, t) q_1 (\bfx, t) + \gy_2 (\bfx, t) q_2 (\bfx, t)], \\
G_1 (t) &= \int_{\bbR^2}\df{\bfx}\; q_1^2 (\bfx, t), \\
G_2 (t) &= \int_{\bbR^2}\df{\bfx}\; q_2^2 (\bfx, t),
\end{align}
under the assumptions $f_1 = f_2 = 0$ and $d_1 = d_2 = 0$. The distribution of energy and potential enstrophy for each layer in Fourier space is described by the energy spectrum $E(k)$ and the corresponding potential enstrophy spectra $G_1 (k)$ and $G_2 (k)$, that are defined via the bracket notation as 
\begin{align}
E(k) &= -\innerf{\gy_1}{q_1}{k} - \innerf{\gy_2}{q_2}{k}, \\
G_1 (k) &= \innerf{q_1}{q_1}{k}, \\
G_2 (k) &= \innerf{q_2}{q_2}{k}.
\end{align}
All three spectra are positive definite for all wavenumbers $k\in (0, +\infty)$. The minus sign in the definition of $E(k)$ is needed to ensure that $E(k)\geq 0$ \cite{article:Gkioulekas:p16}. In previous work \cite{article:Tung:2007,article:Gkioulekas:p15,article:Gkioulekas:p16} we have also found it useful to define the streamfunction spectra $U_1 (k), U_2 (k), C_{12} (k)$ given by 
\begin{align}
U_1 (k) &= \innerf{\gy_1}{\gy_1}{k}, \\
U_2 (k) &= \innerf{\gy_2}{\gy_2}{k}, \\
C_{12} (k) &= \innerf{\gy_1}{\gy_2}{k},
\end{align}
which do not correspond to any conservation law, but are useful in studying  the dissipation rate spectra for energy and potential enstrophy. Since $U_1 (k) + U_2 (k)\pm 2C_{12}(k) = \innerf{\gy_1\pm\gy_2}{\gy_1\pm\gy_2}{k} \geq 0$, it follows that $C_{12}(k)$ is restricted by an arithmetic-geometric mean inequality 
\begin{equation}
2|C_{12}(k)| \leq U_1 (k)+ U_2 (k),
\end{equation}
whereas both $U_1 (k)$ and $U_2(k)$ are positive-definite and satisfy $U_1 (k)\geq 0$ and $U_2 (k)\geq 0$ over all wavenumbers $k\in (0,+\infty)$.

Our point of departure is the definition of the barotropic energy spectrum $E_K (k)$ and the baroclinic energy spectrum $E_P (k)$ \cite{article:Salmon:1980,article:Tung:2007}. Let $\gy = (\gy_1 + \gy_2)/2$ and $\gt = (\gy_1-\gy_2)/2$ and note that $\gy_1 = \gy + \gt$ and $\gy_2 = \gy-\gt$. It follows that
\begin{align}
q_1 &= \lapl (\gy+\gt) - k_R^2 \gt, \\
q_2 &= \lapl (\gy-\gt) + k_R^2 \gt,
\end{align}
and therefore
\begin{align}
\innerf{\gy_1}{q_1}{k} &= \innerf{\gy+\gt}{\lapl (\gy+\gt) - k_R^2 \gt}{k}  \\
&= \innerf{\gy}{\lapl\gy}{k} +  \innerf{\gy}{\lapl\gt}{k} +  \innerf{\gt}{\lapl\gy}{k} +  \innerf{\gt}{\lapl\gt}{k} - k_R^2 \innerf{\gy+\gt}{\gt}{k},
\label{eq:energy-spec-term-one}
\end{align}
and
\begin{align}
\innerf{\gy_2}{q_2}{k} &= \innerf{\gy-\gt}{\lapl (\gy-\gt) + k_R^2 \gt}{k} \\
&= \innerf{\gy}{\lapl\gy}{k} -  \innerf{\gy}{\lapl\gt}{k} -  \innerf{\gt}{\lapl\gy}{k} +  \innerf{\gt}{\lapl\gt}{k} - k_R^2 \innerf{\gt-\gy}{\gt}{k}.
\label{eq:energy-spec-term-two}
\end{align}
Adding Eq.~\eqref{eq:energy-spec-term-one} and Eq.~\eqref{eq:energy-spec-term-two} gives the energy spectrum $E(k)$ which simplifies to
\begin{align}
E(k) &= -\innerf{\gy_1}{q_1}{k}-\innerf{\gy_2}{q_2}{k} 
= -2\innerf{\gy}{\lapl\gy}{k} -2\innerf{\gt}{\lapl\gt}{k} + 2k_R^2 \innerf{\gt}{\gt}{k} \\
&= 2k^2 \innerf{\gy}{\gy}{k} + 2(k^2 + k_R^2)\innerf{\gt}{\gt}{k}.
\end{align}
The $\psi$-dependent term corresponds to the barotropic energy spectrum $E_K (k)$ and the $\gt$-dependent term corresponds to the baroclinic energy spectrum $E_P (k)$. Consequently, we define:
\begin{align}
E_K (k) &= 2k^2 \innerf{\gy}{\gy}{k}, \\
E_P (k) &= 2(k^2 + k_R^2)\innerf{\gt}{\gt}{k}.
\end{align}
Following Salmon \cite{article:Salmon:1980} and previous work \cite{article:Tung:2007}, it is also useful to define 
\begin{equation}
E_C (k) = 2k^2 \innerf{\gy}{\gt}{k}.
\end{equation}

Let us define a function $P(k)$ such that $E_P (k) = P(k) E(k)$ and $E_K (k) = [1-P(k)] E(k)$ with $P(k)$ satisfying the inequality $0\leq P(k)\leq 1$. Following the phenomenology proposed by Salmon \cite{article:Salmon:1978,article:Salmon:1980,book:Salmon:1998}, we anticipate that for wavenumbers $k\ll k_R$ near the forcing range we have $P(k) \approx 0$ and that for wavenumbers $k\gg k_R$ beyond the Rossby wavenumber we have $P(k) \approx 1/2$. Going one step further, we can establish a relationship between $E_C(k)$ and $E(k)$ if we introduce an additional function $\gC (k)$ such that $G_1 (k) = \gC (k) G(k)$ and $G_2 (k) = [1-\gC (k)] G(k)$, capturing the distribution of potential enstrophy between the upper and lower layers. It is important to note that the nonlinear terms conserve potential enstrophy for each layer separately and cannot redistribute potential enstrophy between the two layers. However, asymmetric dissipation may remove potential enstrophy at different dissipation rates between the two layers, and result in a variation of $\gC (k)$ with increasing wavenumbers. 

First, we observe that $G_1 (k)$ and $G_2 (k)$ can be written in terms of $E_K (k), E_P (k), E_C (k)$ as follows:
\begin{align}
G_1 (k) &= \innerf{q_1}{q_1}{k} \\
&= \innerf{\lapl (\gy+\gt) - k_R^2 \gt}{\lapl (\gy+\gt) - k_R^2 \gt}{k} \\
&= \innerf{\lapl\gy}{\lapl\gy}{k} + 2\innerf{\lapl\gy}{\lapl\gt - k_R^2 \gt}{k} + \innerf{\lapl\gt - k_R^2 \gt}{\lapl\gt - k_R^2 \gt}{k} \\
&= k^4 \innerf{\gy}{\gy}{k} + 2k^2 (k^2+k_R^2) \innerf{\gy}{\gt}{k} + (k^2 + k_R^2)^2 \innerf{\gt}{\gt}{k} \\
&= (1/2) k^2 E_K (k) + (1/2)(k^2+k_R^2) E_P (k) + (k^2+k_R^2) E_C (k),
\label{eq:potential-enstr-spec-one}
\end{align}
and
\begin{align}
G_2 (k) &= \innerf{q_2}{q_2}{k} \\
&= \innerf{\lapl (\gy-\gt) + k_R^2 \gt}{\lapl (\gy-\gt) + k_R^2 \gt}{k} \\
&= \innerf{\lapl\gy}{\lapl\gy}{k} + 2\innerf{\lapl\gy}{-\lapl\gt + k_R^2 \gt}{k} + \innerf{-\lapl\gt + k_R^2 \gt}{-\lapl\gt + k_R^2 \gt}{k} \\
&= k^4 \innerf{\gy}{\gy}{k} - 2k^2 (k^2+k_R^2) \innerf{\gy}{\gt}{k} + (k^2 + k_R^2)^2 \innerf{\gt}{\gt}{k} \\
&= (1/2) k^2 E_K (k) + (1/2)(k^2+k_R^2) E_P (k) - (k^2+k_R^2) E_C (k).
\label{eq:potential-enstr-spec-two}
\end{align}
Adding Eq.~\eqref{eq:potential-enstr-spec-one} and Eq.~\eqref{eq:potential-enstr-spec-two} gives a relationship between the potential enstrophy spectrum $G (k)$ and the energy spectrum  $E (k)$, in terms of the function $P (k)$:
\begin{align}
G (k) &= G_1 (k) + G_2 (k) = k^2 E_K (k) + (k^2 + k_R^2) E_P (k) \\
&= [k^2 (1-P (k))+(k^2+k_R^2) P (k)] E (k),
\end{align}
which is in turn used to calculate the spectrum $E_c (k)$ in terms of $E (k)$, $P (k)$, and $\gC (k)$ by noting that if we instead subtract Eq.~\eqref{eq:potential-enstr-spec-two} from Eq.\eqref{eq:potential-enstr-spec-one}, we obtain:
\begin{equation}
G_1 (k) - G_2 (k) = 2(k^2+k_R^2) E_C (k),
\end{equation}
and solving for $E_C (k)$ gives:
\begin{align}
E_C (k) &= \frac{G_1 (k) - G_2 (k)}{2(k^2+k_R^2)} 
= \frac{\gC (k) G(k)-[1-\gC (k)] G(k)}{2(k^2+k_R^2)} 
= \frac{[2\gC (k)-1] G(k)}{2(k^2+k_R^2)} \\
&= \frac{[2\gC (k)-1][k^2 (1-P(k))+(k^2+k_R^2) P(k)]}{2(k^2+k_R^2)} E(k) 
= \frac{[2\gC (k)-1][k^2+k_R^2 P(k)] E(k)}{2(k^2+k_R^2)}.
\end{align}
Finally, given the relationships between the spectra $E_K (k)$, $E_P (k)$, $E_C (k)$ and the total energy spectrum $E(k)$ as well as the functions $P(k)$ and $\gC (k)$, we are also able to express the streamfunction spectra $U_1 (k)$, $U_2 (k)$, and $C_{12} (k)$ in terms of $E(k)$, $P(k)$, $\gC (k)$, as shown in the following:
\begin{align}
U_1 (k) &= \innerf{\gy_1}{\gy_1}{k}  = \innerf{\gy+\gt}{\gy+\gt}{k}  
= \innerf{\gy}{\gy}{k}  +2\innerf{\gy}{\gt}{k}  + \innerf{\gt}{\gt}{k}  \\
&= \frac{E_K (k)}{2k^2}+\frac{E_C (k)}{k^2}+\frac{E_P (k)}{2(k^2+k_R^2)}
= \frac{(k^2+k_R^2) E_K (k) + k^2 E_P (k)+2(k^2+k_R^2) E_C (k)}{2k^2(k^2+k_R^2)} \\
&= \frac{(k^2+k_R^2)[1-P(k)]+k^2 P(k)+[2\gC (k)-1][k^2+k_R^2 P(k)]}{2k^2(k^2+k_R^2)} E(k),
\label{eq:U-one-spectrum}
\end{align}
and
\begin{align}
U_2 (k) &= \innerf{\gy_2}{\gy_2}{k}  = \innerf{\gy-\gt}{\gy-\gt}{k}  
= \innerf{\gy}{\gy}{k}  - 2\innerf{\gy}{\gt}{k}  + \innerf{\gt}{\gt}{k} \\
&= \frac{E_K (k)}{2k^2}-\frac{E_C (k)}{k^2}+\frac{E_P (k)}{2(k^2+k_R^2)}
= \frac{(k^2+k_R^2) E_K (k) + k^2 E_P (k)-2(k^2+k_R^2) E_C (k)}{2k^2(k^2+k_R^2)} \\
&= \frac{(k^2+k_R^2)[1-P(k)]+k^2 P(k)-[2\gC (k)-1][k^2+k_R^2 P(k)]}{2k^2(k^2+k_R^2)} E(k),
\label{eq:U-two-spectrum}
\end{align}
and
\begin{align}
C_{12}(k) &= \innerf{\gy_1}{\gy_2}{k}  = \innerf{\gy+\gt}{\gy-\gt}{k}  
= \innerf{\gy}{\gy}{k} - \innerf{\gt}{\gt}{k} \\
&= \frac{E_K (k)}{2k^2} - \frac{E_P (k)}{2(k^2+k_R^2)}
= \frac{(k^2+k_R^2) E_K (k) - k^2 E_P (k)}{2k^2(k^2+k_R^2)} \\
&= \frac{(k^2+k_R^2)[1-P(k)] - k^2 P(k)}{2k^2(k^2+k_R^2)} E(k).
\label{eq:C-one-two-spectrum}
\end{align}
Our first major result is an immediate consequence of the fact that the streamfunction spectra $U_1 (k)$ and $U_2 (k)$ are unconditionally positive over all wavenumbers and it establishes a restriction between $P (k)$ and $\gC (k)$, given by the following proposition:
\begin{proposition}
For all wavenumbers $k\in (0,+\infty)$, we have:
\begin{equation}
|2\gC (k)-1| \leq \frac{k^2+[1-P (k)] k_R^2}{k^2+k_R^2 P (k)}.
\label{eq:distribution-constraint}
\end{equation}
\label{prop:distribution-constraint}
\end{proposition}

\begin{proof}
Since $U_1 (k) = \innerf{\gy_1}{\gy_1}{k} \geq 0$, and $E(k) \geq 0$, and $k^2 (k^2+k_R^2) > 0$ for all wavenumbers $k\in (0, +\infty)$, it follows from Eq.~\eqref{eq:U-one-spectrum} that
\begin{align}
&(k^2+k_R^2)[1-P(k)] + k^2 P(k) + [2\gC (k)-1][k^2+k_R^2 P(k)] \geq 0 \\
\implies& k^2 + k_R^2 [1-P(k)] + [2\gC (k)-1][k^2+k_R^2 P(k)] \geq 0 \\
\implies& [2\gC (k)-1][k^2+k_R^2 P(k)]  \geq -\{k^2 + k_R^2 [1-P(k)]]\}.
\end{align}
Noting that $P(k)\geq 0$ implies that $k^2+k_R^2 P(k) >0$, it follows that
\begin{equation}
2\gC (k)-1 \geq \frac{-\{k^2+[1-P(k)] k_R^2\}}{k^2+k_R^2 P(k)}.
\label{eq:distribution-constraint-eq-one}
\end{equation}
Likewise, since $U_2 (k) = \innerf{\gy_2}{\gy_2}{k} \geq 0$, via a similar argument with Eq.~\eqref{eq:U-two-spectrum}, we have
\begin{align}
&(k^2+k_R^2)[1-P(k)] + k^2 P(k) - [2\gC (k)-1][k^2+k_R^2 P(k)] \geq 0 \\
\implies& [2\gC (k)-1][k^2+k_R^2 P(k)] \leq k^2 + k_R^2 [1-P(k)] \\
\implies& 2\gC (k)-1 \leq \frac{k^2+(1-P(k)) k_R^2}{k^2+k_R^2 P(k)}.
\label{eq:distribution-constraint-eq-two}
\end{align}
Combining Eq.~\eqref{eq:distribution-constraint-eq-one} and Eq.~\eqref{eq:distribution-constraint-eq-two} proves the claim
\begin{equation}
|2\gC (k)-1| \leq \frac{k^2+[1-P (k)]) k_R^2}{k^2+k_R^2 P (k)}.
\end{equation}
\end{proof}

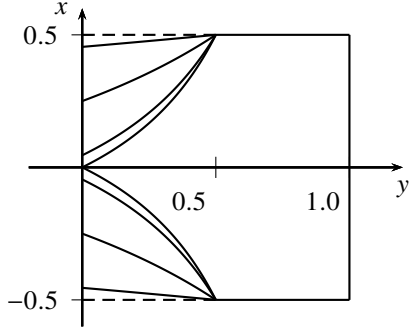
\begin{figure}[tb]\begin{center}
\psset{unit=100pt,labelsep=5pt,algebraic=true}
\begin{pspicture}[shift=*](-0.4,-0.7)(1.3,0.7)
\psaxes[Dx=0.5,Dy=0.5,xlabelOffset=-0.1]{->}(0,0)(-0.2,-0.6)(1.2,0.6)[$y$,-90][$x$,180]
\psplot[plotstyle=curve,plotpoints=200]{0}{0.5}{x/(2*(1-x))}
\psplot[plotstyle=curve,plotpoints=200]{0}{0.5}{-x/(2*(1-x))}
\psline{-}(0.5,0.5)(1,0.5)
\psline(0.5,-0.5)(1,-0.5)
\psline[linestyle=dashed](0,0.5)(0.5,0.5)
\psline[linestyle=dashed](0,-0.5)(0.5,-0.5)
\psline(1,0.5)(1,-0.5)
\psplot[plotstyle=curve,plotpoints=200]{0}{0.5}{(0.1+x)/(2*(0.1+1-x))}
\psplot[plotstyle=curve,plotpoints=200]{0}{0.5}{-(0.1+x)/(2*(0.1+1-x))}
\psplot[plotstyle=curve,plotpoints=200]{0}{0.5}{(1+x)/(2*(1+1-x))}
\psplot[plotstyle=curve,plotpoints=200]{0}{0.5}{-(1+x)/(2*(1+1-x))}
\psplot[plotstyle=curve,plotpoints=200]{0}{0.5}{(10+x)/(2*(10+1-x))}
\psplot[plotstyle=curve,plotpoints=200]{0}{0.5}{-(10+x)/(2*(10+1-x))}
\end{pspicture}
\caption{\label{fig:BE-boxes}\small This graph shows the region constraining $\gC (k)$ and $P(k)$ using the representation $\gC (k) = 1/2+x$ and $P(k) = 1-y$. The smallest ``pointy box'' corresponds to the limit $k\ll k_R$. The larger boxes correspond to the wavenumber ratios $k/k_R = 10^{-1/2}, 1, 10^{1/2}$, with the box area increasing with the ratio $k/k_R$. For $k/k_R < 0.1$ the box becomes graphically indistinguishable from the limit $k\ll k_R$. The ``pointy box'' becomes graphically indistinguishable from the limit $k\gg k_R$, where it converges to a square,  when $k/k_R > 10$}
\end{center}\end{figure}

It should be noted that just from the basic restriction $0 \leq \gC (k) \leq 1$, we have $|2\gC (k)-1| \leq 1$. In the limit $k \gg k_R$, the inequality given by Eq.~\eqref{eq:distribution-constraint} also reduces to $|2\gC (k)-1| \leq 1$, and as such, it does not impose any further restrictions on $\gC (k)$. However, in the limit $k\ll k_R$, Eq.~\eqref{eq:distribution-constraint} reduces to $|2\gC (k)-1| \leq (1-P(k))/P(k)$, and with no loss of generality we have:
\begin{equation}
|2\gC (k)-1| \leq \min\left\{ 1, \frac{1-P(k)}{P(k)} \right\}, \quad \text{for } k\ll k_R.
\end{equation}
This restriction corresponds to the smallest ``pointy box'', shown in Fig.~\ref{fig:BE-boxes}, using the representation $\gC(k) = 1/2+x$ and $P(k)=1-y$. For wavenumbers $k/k_R < 0.1$, the region will expand, but the expansion is too small to be seen graphically, consequently, Eq. (73) is a pretty good approximation for any wavenumbers $k/k_R < 0.1$, corresponding to both the forcing range as well as the range of wavenumbers where we expect to see the downscale potential enstrophy cascade energy spectrum scaling $k^{-3}$. 

From a physical point of view, this restriction is very important. If we assume that at the forcing range, all of the injected energy is injected as baroclinic energy, then we expect that \emph{initially} $P(k) \approx 1$, which, in turn, implies that $\gC (k) \approx 1/2$, corresponding to an equipartition of the potential enstrophy between the upper and lower layers. This implies that the baroclinic energy injection at wavenumbers $k\ll k_R$ has to be accompanied with a symmetric injection of potential enstrophy to both layers,  resulting in an initial equipartition of potential enstrophy, in the potential enstrophy spectrum, between the upper layer and the lower layer. This restriction on $\gC(k)$ is weakened, when the flow reaches steady state and the energy is redistributed in the limit $k\ll k_R$ to become predominantly barotropic,  allowing the distribution of potential enstrophy between layers to deviate from exact equipartition. In Section 5, we will show that the asymmetric Ekman term tends to stabilize the fixed point distribution of potential enstrophy between layers,  with more potential enstrophy concentrated in the upper layer than the lower layer.

\section{Dissipation rate spectra for asymmetric Ekman term}

The two-layer quasi-geostrophic model's nonlinear terms conserve the total energy $E$ as well as the total potential enstrophy $G_1$ and $G_2$ in the upper and lower layers. The corresponding conservation laws are written in terms of the time derivative of the energy spectrum $E(k)$ and the potential enstrophy spectra $G_1 (k)$ and $G_2 (k)$ and they read
\begin{align}
\pderiv{E(k)}{t}&+\pderiv{\Pi_E (k)}{k} = -D_E (k) + F_E (k), \\
\pderiv{G_1(k)}{t}&+\pderiv{\Pi_{G_1} (k)}{k} = -D_{G_1} (k) + F_{G_1} (k), \\
\pderiv{G_2(k)}{t}&+\pderiv{\Pi_{G_2} (k)}{k} = -D_{G_2} (k) + F_{G_2} (k).
\end{align}
Here $\Pi_E (k)$ represents the energy flux from the $(0, k)$ interval to the $(k, +\infty)$ interval via the nonlinear term; $D_E (k)$ represents the energy dissipation rate spectrum accounting for the removal of energy via the dissipation terms; $F_E (k)$ represents the energy forcing spectrum accounting for the injection of energy via the forcing terms. Similar definitions apply to $\Pi_{G_1} (k)$, $D_{G_1} (k)$, $F_{G_1} (k)$ for the upper layer potential enstrophy conservation law and to $\Pi_{G_2} (k)$, $D_{G_2} (k)$, $F_{G_2} (k)$ for the lower layer potential enstrophy conservation law. The conservation laws themselves are accounted for via the boundary conditions
\begin{align}
&\Pi_E (0) = \Pi_{G_1} (0) = \Pi_{G_2} (0) = 0, \\
&\lim_{k\to +\infty} \Pi_E (k) = \lim_{k\to +\infty} \Pi_{G_1} (k) = \lim_{k\to +\infty} \Pi_{G_2} (k) = 0.
\end{align}
Our main interest here is to understand the contribution of the asymmetric Ekman term to the dissipation rate spectra $D_E (k)$, $D_{G_1}(k)$, and $D_{G_2}(k)$.

To that end, we begin with a previous general result \cite{article:Gkioulekas:p16} for the dissipation rate spectra for a generalized multilayer model of the form
\begin{equation}
\pderiv{q_\ga}{t}+J(\gy_\ga, q_\ga) = f_\ga + d_\ga,
\end{equation}
with $\ga\in\{1, 2, \ldots, n\}$ representing the layer index. Given the Fourier transform $\hat{\gy_\ga}$ of the streamfunction field so that
\begin{equation}
\gy_\ga (\bfx, t) = \int_{\bbR^2}\hat{\gy}_\ga (\bfk, t)\exp (i\bfk\cdot\bfx)\;\df{\bfk},
\end{equation}
we assume that the relationship between $q_\ga$ and $\gy_\ga$ takes a general linear form
\begin{equation}
q_\ga (\bfx, t) = \sum_{\gb}\int_{\bbR^2} L_{\ab}(\nrm{\bfk}) \hat{\gy}_\gb (\bfk, t) \exp(i\bfk\cdot\bfx)\;\df{\bfk},
\end{equation}
with the additional assumption $L_{\ab} (\bfk) = L_{\ba} (\bfk)$. We also assume that the dissipation terms $d_\ga$ are given as general linear transforms of the streamfunction fields $\gy_\ga$ so that
\begin{equation}
d_{\ga} (\bfx, t) = \sum_{\gb}\int_{\bbR^2} D_{\ab}(\nrm{\bfk})\hat{\gy}_\gb (\bfk, t)\exp (i\bfk\cdot\bfx)\;\df{\bfk}.
\end{equation}
Under these assumptions, we have shown \cite{article:Gkioulekas:p16} that the dissipation rate spectra $D_E (k)$ and $D_{G_\ga} (k)$ can be expressed in terms of the streamfunction spectra $C_{\ab} (k) = \innerf{\gy_\ga}{\gy_\ga}{k}$ via the equations
\begin{align}
D_E (k) &= 2 \sum_{\ab} D_{\ab} (k) C_{\ab}(k), \label{eq:general-DE-spectrum} \\
D_{G_\ga}(k) &= -2\sum_{\bc} L_{\ab}(k) D_{\ac}(k) C_{\bc}(k). \label{eq:general-DG-spectrum}
\end{align}
Note that in the summation symbols, written above, it is implied that the indices are being summed over all layers. 

For the case of the two-layer quasi-geostrophic model, the function $L_{\ab}(k)$ is a $2\times 2$ square matrix given by
\begin{equation}
L(k) = -\mattwo{a(k)}{b(k)}{b(k)}{a(k)},
\end{equation}
with $a(k)=k^2+k_R^2/2$ and $b(k)=-k_R^2/2$. Likewise, the dissipation term function $D_{\ab}(k)$ is given by
\begin{equation}
D(k)=\mattwo{D_1 (k)}{0}{\mu d(k)}{D_2 (k)+d(k)},
\end{equation}
with $D_1 (k)=\nu k^{2p+2}$, $D_2 (k) = (\nu + \gD\nu) k^{2p+2}$, and $d(k) = \nu_E k^2/(\mu +1)$. Since we are interested only in the contribution of the asymmetric Ekman term to the dissipation rate spectra, we will assume that $D_1(k) = D_2 (k) = 0$ and use instead
\begin{equation}
D(k) = \mattwo{0}{0}{\mu d(k)}{d(k)}.
\end{equation}
Consequently, from Eq.~\eqref{eq:general-DE-spectrum}, the energy dissipation rate spectrum $D_E (k)$ is given by
\begin{align}
D_E (k) &= 2[D_{21} (k) C_{21}(k)+D_{22}(k) C_{22}(k)]
= 2\mu d(k) C_{12}(k)+d(k)U_2 (k) \\
&= [2\mu C_{12}(k)+U_2 (k)] d(k),
\end{align}
noting that the contributions that correspond to $D_{11}(k)$ and $D_{12}(k)$ vanish, because $D_{11}(k)=D_{12}(k)=0$. The potential enstrophy dissipation rate spectrum for the upper layer is zero, because, from Eq.~\eqref{eq:general-DG-spectrum}
\begin{equation}
D_{G_1}(k) = -2[L_{11}(k) D_{11}(k) C_{11}(k) + L_{11}(k) D_{12}(k) C_{12}(k) + L_{12}(k) D_{11}(k) C_{21}(k) + L_{12}(k) D_{12}(k) C_{22}(k)] = 0,
\end{equation}
and we note that all contributions involve $D_{11}(k)$ and $D_{12}(k)$, both of which vanish. It follows that the asymmetric Ekman term conserves potential enstrophy in the upper layer and only dissipates potential enstrophy from the lower layer. The potential enstrophy dissipation rate spectrum for the lower layer is given by 
\begin{align}
D_{G_2}(k) &= -2[L_{21}(k) D_{21}(k) C_{11}(k) + L_{21}(k) D_{22}(k) C_{12}(k) + L_{22}(k) D_{21}(k) C_{21}(k) + L_{22}(k) D_{22}(k) C_{22}(k)] \\
&= 2[b(k) \mu d(k) U_1 (k) + b(k) d(k) C_{12}(k) + a(k) \mu d(k) C_{12}(k) + a(k) d(k) U_2 (k)] \\
&= [2b(k) \mu U_1 (k) + 2a(k) U_2 (k) + (b(k) + \mu a(k)) C_{12}(k) ] d(k).
\end{align}

Our previous investigation of the asymmetric Ekman term \cite{article:Gkioulekas:p15} was inconclusive because no phenomenological assumptions were made, and without making any such assumptions, we have no useful knowledge about the streamfunction spectra $U_1 (k)$, $U_2 (k)$, and $C_{12}(k)$. However, as we have seen in Section~\ref{sec:energy-potential-enstrophy-spectra}, given the function $P(k)$, describing the distribution of energy in the energy spectrum $E(k)$ between baroclinic and barotropic energy, and given the function $\gC (k)$, describing the distribution of potential enstrophy in the potential enstrophy spectra $G_1 (k)$ and $G_2 (k)$ between the upper and lower layers, it is possible to express the streamfunction spectra $U_1 (k)$, $U_2 (k)$, and $C_{12}(k)$ in terms of the energy spectrum $E(k)$ and the functions $P(k)$ and $\gC (k)$ via Eq.~\eqref{eq:U-one-spectrum}, Eq.~\eqref{eq:U-two-spectrum}, and Eq.~\eqref{eq:C-one-two-spectrum}. Substituting these equations to our expressions above for the dissipation rate spectra $D_E (k)$, $D_{G_1}(k)$, and $D_{G_1}(k)$ results in very tedious calculations, for which we have used the open source computer algebra system Maxima \cite{maxima}, leading to the following equations:
\begin{align}
D_E (k) &= \frac{[B_E^{(1)}(k) (k/k_R)^2 + B_E^{(2)}(k)]  k_R^2 d(k) E(k)}{k^2 (k^2+k_R^2)}, \label{eq:DE-equation} \\
D_{G_2} (k) &= \frac{[B_G^{(1)}(k) (k/k_R)^4 +B_G^{(2)}(k) (k/k_R)^2 +B_G^{(3)}(k)] k_R^4 d(k) E(k)}{2k^2 (k^2+k_R^2)}, \label{eq:DG2-equation}
\end{align}
with $B_E^{(1)}(k)$ and $B_E^{(2)}(k)$ non-dimensional coefficients given by
\begin{align}
B_E^{(1)}(k) &= 2[1-\gC (k)]+\mu [1-2P(k)], \\
B_E^{(2)}(k) &= [1-2\gC (k) P(k)]+\mu [1-P(k)],
\end{align}
and likewise with $B_G^{(1)}(k)$, $B_G^{(2)}(k)$, and $B_G^{(3)}(k)$ given by
\begin{align}
B_G^{(1)}(k) &= 4[1-\gC (k)]+2\mu [1-2P(k)], \\
B_G^{(2)}(k) &= [-4\gC(k) P(k) +2P(k) -2\gC (k) +3]+\mu [3-2\gC (k) -4P(k)], \\
B_G^{(3)}(k) &= (\mu+1)[1-2\gC (k)] P(k).
\end{align}
We note that $d(k)>0$ and $E(k)\geq 0$, so most of the physical arguments given in the following are based on determining the signs of the coefficients $B_E^{(1)}(k) $, $B_E^{(2)}(k) $, $B_G^{(1)}(k) $, $B_G^{(2)}(k)$, and $B_G^{(3)}(k)$. Note that since $D_{G_1}(k)=0$, we do not need to concern ourselves with the potential enstrophy dissipation rate of the upper layer.

Using these equations for the dissipation rate spectra $D_E (k)$ and $D_{G_2}(k)$ as a point of departure, we will now try to bring out as much physical insight as we can about the role of the asymmetric Ekman dissipation term in the phenomenology of the two-layer quasi-geostrophic model. 

\section{The potential enstrophy dissipation rate spectrum $D_{G_2}(k)$}
\label{sec:potential-enstrophy-dissipation-rate}

The effect of the asymmetric Ekman term on the potential enstrophy dissipation rate is fundamentally different between the limit $k\ll k_R$, corresponding to the forcing range and part of the observable downscale potential enstrophy cascade, and the limit $k \gg k_R$ corresponding to the observable downscale energy cascade. First, we note that the coefficients $B_G^{(1)}(k)$, $B_G^{(2)}(k)$, and $B_G^{(3)}(k)$ are all bounded over all wavenumbers $k\in (0,+\infty)$, as shown in the following:
\begin{align}
|B_G^{(1)}(k)|  &= |4[1-\gC (k)]+2\mu [1-2P(k)]| \\
&\leq 4|1-\gC (k)|-2\mu |1-2P(k)| \\
&\leq 4-2\mu (1+2) = 4-6\mu, \\
|B_G^{(2)}(k)| &= |[-4\gC(k) P(k) +2P(k) -2\gC (k) +3]+\mu [3-2\gC (k) -4P(k)]| \\
&\leq 4|\gC(k) P(k)| +2|P(k)|+2|\gC (k)| +3 -\mu [3+2|\gC (k)| +4|P(k)|] \\
&\leq 4+2+2+3-\mu (3+2+4) = 11-9\mu, \\
|B_G^{(3)}(k)| &= |(\mu+1)[1-2\gC (k)] P(k)| \\
&= (\mu+1)|1-2\gC (k)| |P(k)| \\
&\leq (\mu + 1)(1+2)1 = 3(\mu +1).
\end{align}
Consequently, in the limit $k\ll k_R$, using Eq.~\eqref{eq:DG2-equation}, the dominant contribution to the potential enstrophy dissipation rate spectrum $D_{G_2} (k)$ is given by 
\begin{align}
D_{G_2}(k) &\sim \frac{k_R^4 B_G^{(3)}(k) d(k) E(k)}{2k^2 k_R^2} \\
&\sim \frac{\mu+1}{2}\fracp{k_R}{k}^2 [1-2\gC (k)] P(k) d(k) E(k), \quad\text{with } k\ll k_R.
\end{align}

Since $\mu\in [-1/3, 0]$, we have $\mu+1>0$, and furthermore $P(k)\geq 0$ and $d(k)>0$ and $E(k)\geq 0$ for all wavenumbers $k\in (0, +\infty)$, it follows that the sign of the leading term contribution to $D_{G_2}(k)$ is controlled exclusively by the factor $1-2\gC (k)$, which is positive when $\gC (k)<1/2$ and negative when $\gC (k)>1/2$. This creates a very interesting dynamic. As we have explained previously, the antisymmetric forcing of the two-layer quasi-geostrophic model injects energy at the wavenumbers $k\ll k_R$ as baroclinic energy, consequently we anticipate that for $k\ll k_R$, we initially have $P(k)\approx 1$, and therefore, via the inequality Eq.~\eqref{eq:distribution-constraint}, $\gC (k)$ is constrained in a very narrow interval around $1/2$. This means that an equal amount of potential enstrophy is injected on both the upper and lower layers, along with the baroclinic energy injection, which cannot be redistributed afterwards by the nonlinear interactions, since the potential enstrophy of each layer is separately conserved, except via the dissipation terms, and in the limit $k\ll k_R$, specifically, the asymmetric Ekman term. When $\gC (k)$ rises above $1/2$ (i.e. more potential enstrophy in the upper layer than in the lower layer at the given wavenumber $k$), then $D_{G_2}(k)$ becomes negative and it actually injects potential enstrophy into the lower layer, thereby decreasing $\gC (k)$ back towards $1/2$. Likewise, when $\gC (k)$ falls below $1/2$, then $D_{G_2}(k)$ becomes positive and removes potential enstrophy from the lower layer. This tends to increase $\gC (k)$ back towards $1/2$. As a result, the initial effect of the asymmetric Ekman term in the limit $k \ll k_R$ is to create a stable fixed point for $\gC (k)$ near $1/2$ where the potential enstrophy dissipation vanishes. This allows the potential enstrophy injected onto both layers to cascade towards large wavenumbers without any dissipative distortion. It also stabilizes the potential enstrophy distribution between the two layers so that it is approximately equipartitioned between the two layers. 

The subleading contribution to the potential enstrophy dissipation rate spectrum $D_{G_2}(k)$ becomes important when the leading contribution is suppressed by the numerical coefficient $1-2\gC (k)$. Furthermore, as the flow reaches steady state and becomes predominantly barotropic,  the factor $P(k)$ approaches (but does not converge to) 0, which further suppresses the leading term,  independently of the $1-2\gC (k)$ factor. The sign of the subleading contribution is controlled by the numerical coefficient $B_G^{(2)}(k)$. In \ref{sec:sign-of-BG-two}, we show that when $\mu = -1/3$, it follows unconditionally, via Proposition~\ref{prop:sign-of-BG-two-prop-one}, that $B_G^{(2)}(k)>0$ for all wavenumbers $k\in (0, +\infty)$. The case $\mu = -1/3$ corresponds to extrapolated Ekman dissipation in which the surface layer is placed at $1$ Atm. For all other cases, we have shown via  Proposition~\ref{prop:sign-of-BG-two-prop-two} that, under the assumption $0\leq \gC (k)<1$, we have:
\begin{align}
\systwo{-1/3 < \mu < 0}{\ds |2\gC (k)-1| \leq \min\left\{1, \frac{1-P(k)}{P(k)}\right\}} &\implies B_G^{(2)}(k) >0, \label{eq:sign-of-BG-two-eq-one} \\
\systwo{\mu = 0}{\ds |2\gC (k)-1| < \min\left\{1, \frac{1-P(k)}{P(k)}\right\}} &\implies B_G^{(2)}(k) >0. \label{eq:sign-of-BG-two-eq-two}
\end{align}
The case $\gC (k)=1$ corresponds to no potential enstrophy in the lower layer at wavenumber $k$, and it is trivial, since, in the absense of any potential enstrophy in the lower layer, the corresponding potential enstrophy dissipation rate spectrum will be zero. Consequently, we are not worried about the strict inequality in the condition for Eq.~\eqref{eq:sign-of-BG-two-eq-two}.  The inequality restriction on $\gC (k)$ on the conditions for Eq.~\eqref{eq:sign-of-BG-two-eq-one} and Eq.~\eqref{eq:sign-of-BG-two-eq-two} is the $k\ll k_R$ limit of the rigorous restriction on $\gC (k)$  established via  Proposition~\ref{prop:distribution-constraint}, so it is expected to be satisfied.  In fact,  as shown in Fig.~\ref{fig:BE-boxes},  the limit $k\ll k_R$ restriction on $\gC (k)$ is graphically indistinguishable from the exact inequality of Proposition 1, when $k/k_R < 0.1$. Furthermore, as the flow becomes increasingly barotropic and $P(k)$ gets closer to $0$, it becomes easier to satisfy the constraint of Eq.~\eqref{eq:sign-of-BG-two-eq-one} and Eq.~\eqref{eq:sign-of-BG-two-eq-two}. In fact, when more than half of the energy is barotropic (i.e. $P(k) < 1/2$), the constraints of Eq.~\eqref{eq:sign-of-BG-two-eq-one} and Eq.~\eqref{eq:sign-of-BG-two-eq-two} are satisfied unconditionally. 

The conclusion $B_G^{(2)}(k)>0$ implies that, if we have exactly $\gC (k)=1/2$ (i.e. the leading contribution is exactly equal to zero), then the positive subleading contribution results in $D_{G_2}(k)>0$. Consequently, the asymmetric Ekman term will dissipate potential enstrophy from the lower layer and tend to be increase $\gC (k)$ above $1/2$. This results in a competition between the leading and the subleading contributions, with the leading contribution being negative and the subleading contribution being positive. The two contributions balance out at some location $\gC (k) = 1/2+\gc_0 (k)$ with $\gc_0 (k)>0$, and that is the more precise location of the stable fixed point in which the potential enstrophy dissipation vanishes. 

The expected barotropization of the flow, as we approach steady state, tends to further suppress the leading contribution to $D_{G_2}(k)$ the via the $P(k)$ factor on $B_G^{(1)}(k)$. In fact, the leading contribution vanishes with $P(k)=0$.  On the other hand, in the extreme limit $P(k)=0$, the coefficient $B_G^{(2)}(k)$ of the subleading term remains finite, because it simplifies to
\begin{equation}
B_G^{(2)}(k) = (1+\mu)(3-2\gC(k)) \geq 1+\mu \geq 2/3.
\end{equation}
Here, we have used $\gC (k)\leq 1$, for the first inequality, and $\mu\geq -1/3$, for the second inequality. Physically,  this means that as the flow becomes increasingly barotropic, in the limit $k\ll k_R$, the location of the stable fixed point should shift towards redistributing more potential enstrophy towards the upper layer. On the other hand, based on Rhines \cite{article:Rhines:1979} we may expect that a typical order of magnitude estimate for $P(k)$ for $k\ll k_R$ is $P(k) \sim 0.1$, which will be counteracted by the $(k/k_R)^2$ factor in front of $B_{G_2}(k)$ in the subleading term in Eq.~\eqref{eq:DG2-equation} that has an even smaller order of magnitude, in the limit $k\ll k_R$, when there is at least one decade separation between the two wavenumbers. Consequently, we anticipate that the shift away from equipartition will be limited.  

Now, let us consider the potential enstrophy dissipation rate spectrum $D_{G_2}(k)$ of the asymmetric Ekman term in the limit $k\gg k_R$. From Eq.~\eqref{eq:DG2-equation}, the leading contribution to $D_{G_2}(k)$ is now given by:
\begin{align}
D_{G_2}(k) &\sim \frac{B_G^{(1)}(k) k^4 d(k) E(k)}{2k^4} \\
&\sim (1/2) B_G^{(1)}(k) d(k) E(k),
\quad \text{with } k \gg k_R,
\end{align}
with $B_G^{(1)}(k)$ given by
\begin{equation}
B_G^{(1)}(k) = 2[1-\gC (k)]+\mu [1-2P(k)]. \label{eq:BG-one-definition}
\end{equation}

For the case of standard Ekman dissipation (i.e. $\mu =0$), the inequality $0\leq\gC(k)<1$ immediately gives $B_G^{(1)}(k)>0$ and therefore $D_{G_2}(k)>0$, meaning that potential enstrophy will be dissipated from the lower layer. Since this will tend to move $\gC (k)$ towards $1$, as the total potential enstrophy becomes increasingly concentrated in the upper layer, the leading contribution to $D_{G_2}(k)$ may be overtaken by the subleading term, whose sign is determined by the numerical coefficient $B_G^{(2)}(k)$. In the limit $k\gg k_R$, we anticipate, according to Salmon's phenomenology \cite{article:Salmon:1978,article:Salmon:1980,book:Salmon:1998}, that approximately half of the energy in the energy spectrum will be barotropic,  with the other half baroclinic, although we cannot conclude whether most of the energy will be barotropic or baroclinic.  As long as at least half of the energy spectrum at the wavenumber $k$ is barotropic, the assumptions of Eq.~\eqref{eq:sign-of-BG-two-eq-one} and Eq.~\eqref{eq:sign-of-BG-two-eq-two} will be satisfied, so we expect that $B_G^{(2)}(k)>0$, which in turn implies that the potential enstrophy dissipation rate spectrum $D_{G_2}(k)$ will remain positive. It follows that,  contrary to the behavior of the Ekman term at small wavenumbers, we expect it to solely dissipate potential enstrophy from the lower layer in the limit $k\gg k_R$. In doing so, it will tend to increase $\gC (k)$ towards a stable fixed point $\gC (k)=1$ where all potential enstrophy becomes concentrated in the upper layer.

The dynamic changes when more than half of the energy in the energy spectrum at $k$ is baroclinic (i.e. $P(k)>1/2$). In this case, due to the limit $k\gg k_R$, the inequality restriction on $\gC (k)$ by Proposition 1 is almost entirely gone, and it is therefore possible to violate the inequality conditions needed for Eq. (108) and Eq. (109). In fact,  Proposition 6 shows that for the case of the standard Ekman term (i.e. $\mu=0$), the coefficient $B_G^{(2)}(k)$ will satisfy $B_G^{(2)}(k)<0$ when $\gC (k)$ is sufficiently close to $1$ if and only if $P(k)>1/2$. Since the coefficient $B_G^{(1)}(k)$ still satisfies $B_G^{(1)}(k)>0$,  it follows that when $B_G^{(2)}(k)<0$, we should expect a competition between the positive leading term and negative subleading term, resulting in shifting the stable fixed point, where potential enstrophy is neither injected nor dissipated, to $\gC (k)$ slightly below $1$. Otherwise, the overall tendency of the asymmetric Ekman term is still to concentrate most of the potential enstrophy on the upper layer.

For the case of extrapolated Ekman dissipation, with $\mu\in [-1/3, 0)$, the sign of $B_G^{(1)}(k)$ can be positive or negative, depending on the values of $\gC (k)$ and $P(k)$. Consequently, it is possible that the Ekman term may be injecting or dissipating potential enstrophy from the lower layer. In \ref{sec:sign-of-BG-one}, we show that in general 
\begin{equation}
\gC (k)<5/6 \implies B_G^{(1)}(k)>0,
\end{equation}
therefore, as long as less than $5/6$ of the potential enstrophy at wavenumber $k$ is concentrated in the upper layer, the Ekman term will dissipate potential enstrophy from the lower layer  thus  tending to increase the concentration of potential enstrophy on the upper layer.  When  $B_G^{(1)}(k)$ is close to zero or negative, consideration needs to be given again to the subleading contribution whose sign is being controlled by $B_G^{(2)}(k)$. According to Proposition~\ref{prop:sign-of-BG-two-prop-one}, when $\mu=-1/3$ (i.e the Ekman term is placed exactly at the surface boundary layer), then we have $B_G^{(2)}(k)>0$ unconditionally, and the stable fixed point, where potential enstrophy is not dissipated, will be placed at $5/6<\gC (k)<1$. According to Proposition~\ref{prop:sign-of-BG-two-prop-three},  when $\mu\in (-1/3,0)$ (i.e. when the Ekman term is placed between the surface boundary layer and the lower layer), then we have $B_G^{(2)}(k)>0$ unconditionally if and only if $1-P(k)>(1+3\mu)/(2+4\mu)$ with the right-hand side of the inequality being some number between $1/2$ and $0$. This means that if the energy at wavenumber $k$   is approximately half barotropic and half baroclinic, then we will have $B_G^{(2)}(k)>0$, and the stable fixed point remains at $5/6<\gC (k)<1$. In the event that the flow is predominantly baroclinic, we can have  $B_G^{(2)}(k)<0$ when $\gC (k)$ is sufficiently close to $1$. The crossover for $P(k)$, where this becomes possible, is between $1/2$ and $1$ with the interval squeezing onto $1$ as the Ekman term is placed closer to the surface boundary layer,   corresponding to increasing the required ratio of baroclinic to barotropic energy. From Eq.~\eqref{eq:BG-one-definition} we see that when $1/2<P(k)<1$,  we will also have unconditionally $B_G^{(1)}(k)\geq\mu (1-2P(k))>0$, which means that it is possible that the leading term could dominate over the subleading term when $\gC (k)=1$, even with $B_G^{(2)}(k)<0$. It  is also possible that the two terms can balance and shift the stable fixed-point where potential enstrophy  dissipation vanishes at $\gC (k)<1$, but still near $1$. Either way,  there is no significant departure in the overall phenomenology of the asymmetric Ekman term. 

In both cases we see that in the limit $k \gg k_R$, the asymmetric Ekman term dissipates potential enstrophy from the lower layer, thereby concentrating most of the potential enstrophy in the upper layer. The only difference between the standard and the extrapolated Ekman term is that the standard Ekman term dissipates potential enstrophy for any $\gC(k)$, when at least half of the energy is barotropic, otherwise the stable fixed point of vanishing potential enstrophy dissipation could shift slightly below $1$. For the extrapolated Ekman term, the stable fixed point is likely to shift below $1$ as well, and if the energy is not predominantly baroclinic then we can bound the fixed point at $5/6<\gC(k)<1$. Aside from this minor difference, we anticipate similar phenomenology in both cases. 

\section{The energy dissipation rate spectrum $D_E (k)$}
\label{sec:energy-dissipation-rate-spectrum}

The effect of the asymmetric Ekman term on the energy dissipation rate spectrum $D_E (k)$ is very obvious for the case of standard Ekman dissipation, where the Ekman term is placed at the lower potential enstrophy layer. Recall that, in general, the energy dissipation rate spectrum $D_E (k)$ is given by
\begin{equation}
D_E (k) = [U_2 (k) + 2\mu C_{12} (k)] d(k).
\end{equation}
For the case of standard Ekman dissipation, we have $\mu=0$, and therefore $D_E  (k) = U_2  (k) d (k)$. Since both $U_2  (k)\geq 0$ and $d (k) \geq 0$, it follows that $D_E  (k) \geq 0$, consequently the asymmetric Ekman term will always dissipate energy, consistently with the physics underlying Ekman friction. As a result, further investigation is not needed. 

For the case of extrapolated Ekman dissipation, where the Ekman term is placed either between the lower potential enstrophy layer and the surface boundary layer or at the surface boundary layer, we have $-1/3 \leq \mu <0$ and because $C_{12}(k)$ can be positive or negative, it is not obvious whether the energy dissipation rate spectrum $D_E(k)$ is always positive. In this section, we shall consider the sign of $D_E(k)$ in the limits $k \ll k_R$ and $k \gg k_R$ in terms of the distribution of energy between baroclinic and barotropic energy per wavenumber $k$ and in terms of the distribution of potential enstrophy between the upper and lower layers. These two parameters, captured by $\gC(k)$ and $P(k)$, define a two-dimensional space, and we will show that over most of the area of that space, $D_E(k)$ is positive, even though there are some small regions where $D_E(k)$ could be negative. 

We begin by noting that the coefficients $B_E^{(1)} (k)$ and $B_E^{(2)} (k)$ are bounded over all wavenumbers $k\in (0, +\infty)$, as shown in the following:
\begin{align}
|B_E^{(1)} (k)| &= |2[1-\gC (k)]+\mu [1-2P (k)]| \\ 
&\leq 2|1-\gC  (k)|-\mu |1-2P (k)| \\
&\leq 2[1+|\gC  (k)|]-\mu [1+2|P (k)|] \\ 
&\leq 2(1+1)-\mu (1+2) = 4-3\mu,
\end{align}
and
\begin{align}
|B_E^{(2)} (k)| &= |[1-2\gC  (k) P (k)]+\mu [1-P (k)]| \\ 
&\leq |1-2\gC  (k) P (k)|-\mu |1-P (k)| \\
&\leq 1+2|\gC  (k) P (k)|-\mu (1+|P (k)|) \\ 
&\leq 1+2-\mu (1+1) = 3-2\mu.
\end{align}
From Eq.~\eqref{eq:DE-equation}, it follows that in the limit $k\ll k_R$, the leading contribution in the energy dissipation rate spectrum $D_E  (k)$ is given by
\begin{equation}
D_E  (k) \sim \frac{B_E^{(2)} (k) d (k)E (k)}{k^2} \quad\text{with } k\ll k_R,
\end{equation}
whereas, in the limit $k\gg k_R$, the leading contribution is instead given by
\begin{equation}
D_E  (k) \sim \frac{B_E^{(1)} (k) d (k) E (k)}{k^2} \quad\text{with } k\gg k_R.
\end{equation}
It is therefore relevant to determine whether the coefficients $B_E^{(1)} (k)$ and $B_E^{(2)} (k)$ are positive or negative. Similarly to the arguments that we used for the corresponding coefficients controlling the potential enstrophy dissipation rate spectrum (see \ref{sec:sign-of-BG-two}), for any given wavenumber $k$ we write $\gC  (k)=1/2+x$ and $P (k)=1-y$ with $x\in [-1/2, 1/2]$ and $y\in [0, 1]$. Then $B_E^{(1)} (k)$ and $B_E^{(2)} (k)$ simplify to:
\begin{align}
B_E^{(1)} (k) &= 2[1-\gC (k)]+\mu [1-2P (k)] \\ 
&= 1-\mu -2x+2\mu y,
\end{align}
and
\begin{align}
B_E^{(2)} (k) &= [1-2\gC (k) P (k)] + \mu [1-P (k)] \\ 
&= y-2x+2xy+\mu y.
\end{align}
We also note that the restriction
\begin{equation}
|2\gC (k)-1| \leq \min\left\{1, \frac{1-P (k)}{P (k)}\right\},
\end{equation}
which follows from Eq.~\eqref{eq:distribution-constraint} in the limit $k\ll k_R$ can be rewritten in terms of $x,y$ to read:
\begin{equation}
|2x| \leq\min\left\{1,\frac{y}{1-y}\right\}.
\label{eq:distribution-constraint-xy}
\end{equation}

\begin{figure}[tb]\begin{center}
\psset{unit=100pt,labelsep=5pt,algebraic=true}
\begin{pspicture}[shift=*](-0.4,-0.7)(1.3,0.7)
\psaxes[Dx=0.5,Dy=0.5,xlabelOffset=-0.1]{->}(0,0)(-0.2,-0.6)(1.2,0.6)[$y$,-90][$x$,180]
\psplot[plotstyle=curve,plotpoints=200]{0}{0.5}{x/(2*(1-x))}
\psplot[plotstyle=curve,plotpoints=200]{0}{0.5}{-x/(2*(1-x))}
\psline{-}(0.5,0.5)(1,0.5)
\psline(0.5,-0.5)(1,-0.5)
\psline[linestyle=dashed](0,0.5)(0.5,0.5)
\psline[linestyle=dashed](0,-0.5)(0.5,-0.5)
\psline(1,0.5)(1,-0.5)
\psline(0.5,0.5)(1,0.33)
\rput(0.75,0.25){$+$}
\rput(0.9,0.42){$-$}
%
\end{pspicture}
\caption{\label{fig:BE-one}\small This figure shows the graph of the equation $B_E^{(1)} (k)=0$, with $\mu=-1/3$, where the coefficient $B_E^{(1)} (k)$ has the crossover from positive to negative, as well as the graph of the equation $|2x| = \min\{1, y/(1-y)\}$ corresponding to the boundary given by Eq.~\eqref{eq:distribution-constraint-xy}. The $B_E^{(1)}(k)$ term dominates the energy dissipation rate spectrum $D_E(k)$ in the limit $k\gg k_R$.}
\end{center}\end{figure}
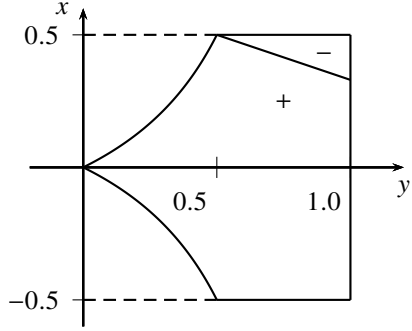

In Fig.~\ref{fig:BE-one}, we display the sign of the coefficient $B_E^{(1)} (k)$ in terms of $x$ and $y$. The ``pointy box'',  passing through the origin and the points $(x,y)\in\{(1/2, 1/2), (1/2, 1), (-1/2, 1), (-1/2, 1/2)\}$ encompasses the region that satisfies the constraint given by Eq.~\eqref{eq:distribution-constraint-xy}. This constraint is rigorous in the limit $k\ll k_R$, or equivalently the limit $k_R\to +\infty$, following from Eq.~\eqref{eq:distribution-constraint}. More generally, for finite Rossby wavenumber $k_R$, the curved part of the pointy box's boundary retreats towards smaller $y$, thereby expanding the area covered by the box, towards becoming a rectangle covering all $(x,y)\in [-1/2, 1/2]\times [-1, 1]$ in the opposite limit $k\gg k_R$ (see Fig.~\ref{fig:BE-boxes}). Since the sign of $B_E^{(1)}(k)$ becomes relevant in the limit $k\gg k_R$, the restriction  of Eq.~\eqref{eq:distribution-constraint-xy} is not relevant, and the distributions of energy and potential enstrophy are unrestricted and can be placed anywhere within the expanded box.

Fig.~\ref{fig:BE-one} also shows the curve defined by the equation $B_E^{(1)} (k)=0$, which is, as a matter of fact, a straight line, since:
\begin{align}
B_E^{(1)} (k)=0 &\ifonlyif 1-\mu -2x+2\mu y = 0 \\
&\ifonlyif x = \frac{1-\mu}{2}+\mu y.
\end{align}
The line passes through the points $(x, y) = (1/2, 1/2)$ and $(x, y) = ((1+\mu)/2, 1)$, with both points on the boundary of the pointy box. Since for $x=y=0$ it is obvious that $B_E^{(1)} (k)$ is positive, we expect that in general $B_E^{(1)} (k)$ is positive below the line, encompassing most of the area of the pointy box, as well as the expanded box, and $B_E^{(1)} (k)$ is negative only in the very small slice above the line. With $\mu\in [-1/3, 0)$, as $\mu$ approaches $0$, the line tends to become horizontal, with the small negative slice vanishing when $\mu = 0$. 

The sign of the energy dissipation rate spectrum $D_E (k)$ approaches the sign of $B_E^{(1)}(k)$ in the limit $k\gg k_R$, where we anticipate that the energy is approximately half barotropic and half baroclinic, and most of the potential enstrophy in the lower layer has been dissipated,  corresponding to $(x,y)$ near $(1/2, 1/2)$, placing us in the vicinity of the negative region. We see that if the flow is at least $1/2$ baroclinic at steady state,  then we're going to be outside the negative region and the asymmetric Ekman term will dissipate total energy.  If the flow is instead predominantly barotropic, then the asymmetric Ekman term can inject total energy,  if the potential enstrophy is also mostly concentrated in the upper layer. The threshold for the required concentration of potential enstrophy  in the upper layer tends to decrease as the Ekman term is placed closer to the surface boundary layer. 

Furthermore, we can predict that intensifying the asymmetric Ekman term, by increasing the coefficient $\nu_E$, will increase the concentration of potential enstrophy in the upper layer, thus pushing $x$ towards $1/2$. As long as most of the energy is barotropic (i.e. $y>1/2$), we can expect that increasing $\nu_E$ will tend to push us inside the negative strip for sufficiently large $\nu_E$. One may  ask whether the negative region is stable or unstable, meaning whether the dynamic of the Ekman term will tend to attract us towards the negative region or push us away from it. Intuitively, we anticipate that $D_E (k)<0$ would correspond to injection of barotropic energy at a greater rate than a corresponding dissipation of baroclinic energy, implying a tendency to increase $y$ and push us towards the upper right direction and further into the negative region. 

\begin{figure}[tb]\begin{center}
\psset{unit=100pt,labelsep=5pt,algebraic=true}
\begin{pspicture}[shift=*](-0.4,-0.7)(1.3,0.7)
\psaxes[Dx=0.5,Dy=0.5,xlabelOffset=-0.1]{->}(0,0)(-0.2,-0.6)(1.2,0.6)[$y$,-90][$x$,180]
\psplot[plotstyle=curve,plotpoints=200]{0}{0.5}{x/(2*(1-x))}
\psplot[plotstyle=curve,plotpoints=200]{0}{0.5}{-x/(2*(1-x))}
\psline{-}(0.5,0.5)(1,0.5)
\psline(0.5,-0.5)(1,-0.5)
\psline[linestyle=dashed](0,0.5)(0.5,0.5)
\psline[linestyle=dashed](0,-0.5)(0.5,-0.5)
\psline(1,0.5)(1,-0.5)
\psplot[plotstyle=curve,plotpoints=200]{0}{0.6}{x/(3*(1-x))}
\rput(0.75,0.25){$+$}
\rput(0.52,0.45){$-$}
\end{pspicture}
\caption{\label{fig:BE-two}\small This figure shows the graph of the equation $B_E^{(2)} (k)=0$, with $\mu=-1/3$,  where $B_E^{(2)} (k)$ has the crossover from positive to negative and the graph of the equation $|2x| = \min\{1, y/(1-y)\}$ corresponding to the boundary of the restriction given by Eq.~\eqref{eq:distribution-constraint-xy}.  The $B_E^{(2)}(k)$ term dominates the energy dissipation rate spectrum $D_E(k)$ in the limit $k\ll k_R$.}
\end{center}\end{figure}
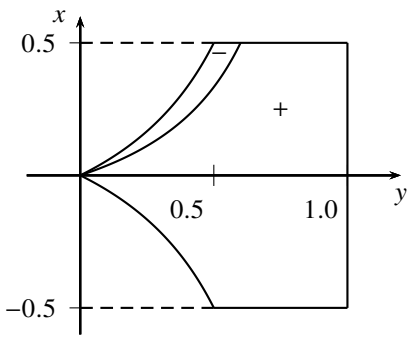

Now, let us consider the coefficient $B_E^{(2)} (k)$ which is relevant to the sign of the energy dissipation rate spectrum $D_E  (k)$ in the limit $k\ll k_R$. A possible cross-over from positive sign to negative sign occurs along the curve given by
\begin{align}
B_E^{(2)} (k) = 0 &\ifonlyif y-2x+2xy+\mu y = 0 \\
&\ifonlyif x = \frac{(\mu+1)y}{2(1-y)}.
\end{align}
We display this curve along with the pointy box given by Eq.~\eqref{eq:distribution-constraint-xy} in Fig.~\ref{fig:BE-two}. Since it is obvious that $B_E^{(2)} (k)$ is positive when $x=0$ and $y=1$, it follows that in the big region below the curve $B_E^{(2)} (k)=0$, we expect that $B_E^{(2)} (k)>0$ and in the small region above the curve $B_E^{(2)} (k)=0$, we expect that $B_E^{(2)} (k)<0$. About the curve, we note that it passes through the points $(x, y) = (0, 0)$ and $(x,y) = (1/2, 1/(\mu +2))$, with the second point corresponding with the intersection of the curve with the ``pointy box''. Since $\mu < 0$, it follows that $1/(\mu +2) > 1/2$, therefore the intersection occurs at the horizontal part of the pointy box where $x > 1/2$. Furthermore, in \ref{app:fig-two-geometry}, we show that the cross-over curve $B_E^{(2)} (k)=0$ is always below the curve $2x = y/(1-y)$ tracing the curvy part of the ``pointy box'', for any $\mu\in [-1/3, 0)$ over the interval $0<x\leq 1/2$. At $x = 1/2$, the pointy box boundary transitions into a horizontal line segment and the cross-over curve continues to remain below the ``pointy box'' boundary until they intersect at $x = 1/(\mu +2)$. For the case of standard Ekman dissipation (i.e. $\mu = 0$), the negative region of the coefficient $B_E^{(2)}(k)$ vanishes. 

In the limit $k\ll k_R$, as we discussed in the previous section, we are expecting that most of the energy is barotropic, corresponding to $x$ near $1$, and that the distribution of potential enstrophy between the two layers settles down on a fixed point with more potential enstrophy in the upper layer than the lower layer. Consequently,  the asymmetric Ekman term will have a positive dissipation rate spectrum $D_E (k)>0$ and will maintain the expected Salmon  phenomenology \cite{article:Salmon:1978,article:Salmon:1980,book:Salmon:1998} in the limit $k\ll k_R$, where an inverse barotropic energy cascade is expected to be dissipated by the Ekman term at small wavenumbers. In fact we see from Fig.~\ref{fig:BE-two} that the asymmetric Ekman term remains dissipative, as long as at least $1/2$ of the energy in the limit $k\ll k_R$ is barotropic, so, for most realistic parameterizations, it is very unlikely that we will have a distribution of energy and potential enstrophy placing us inside the narrow negative region of Fig.~\ref{fig:BE-two}. In the limit of very large Ekman coefficient $\nu_E$ we expect the flow at $k\ll k_R$ to become increasingly baroclinic, so it is not theoretically impossible to find ourselves inside the negative region. On the other hand,  we anticipate that if that were to happen,  then for $D_E (k)<0$ the tendency of the Ekman term will be to inject barotropic energy at a greater rate than it dissipates baroclinic energy,  and with the nonlinear interactions also having a weakened tendency to convert energy from baroclinic to barotropic, we would expect that the steady state solution would tend to be repelled by the negative region of Fig.~\ref{fig:BE-two}. 

In connection with the foregoing analysis, we should emphasize an important observation: comparing the negative regions for $B_E^{(1)} (k)$ and $B_E^{(2)} (k)$, as shown in Fig.~\ref{fig:BE-one} and Fig.~\ref{fig:BE-two}, they hardly ever have any overlap except when $y = 1/2$ and when $x$ is very near $1/2$. There is a very small spot there where both coefficients are negative, which is sufficient for a negative energy dissipation rate spectrum $D_E  (k)<0$. A consequence of this observation is that for finite wavenumbers $k<k_R$ (as opposed to the limit $k\ll k_R$, i.e. $k_R \to +\infty$), being in the negative region of $B_E^{(2)} (k)$ is necessary but not sufficient for having a negative energy dissipation rate spectrum $D_E  (k)$, because the subleading term, controlled by $B_E^{(1)} (k)$ is going to be positive, so the actual region where $D_E  (k)$ is negative is, in fact, a subset of the negative region shown in Fig.~\ref{fig:BE-two}. For finite wavenumbers $k>k_R$, a similar argument applies, from which we conclude that the actual region where $D_E  (k)$ is negative is a subset of the negative region shown in Fig.~\ref{fig:BE-one}.

Overall, we expect that the asymmetric Ekman term will dissipate energy, both in its standard form and in the extrapolated form. We have seen that the negative regions where the energy dissipation rate spectrum may become negative are very small to begin with, and being within these negative regions is not even sufficient to ensure that the energy dissipation rate spectrum will be negative. We have also  seen that the expected phenomenology for $k\ll k_R$ will tend to drive the system away from the corresponding negative region, whereas for $k\gg k_R$, it will tend to drive the system into the negative region, if the energy is predominantly barotropic.

\section{Conclusion and Discussion}

A fundamental difference between the two-layer quasi-geostrophic model and the two-dimensional Navier-Stokes equations is that in the former, there is a wider variety in the possible configurations of the dissipation terms, resulting in different behaviors in the dissipation rate spectra of energy and potential enstrophy. These can have a substantial effect in the phenomenology of the two-layer quasi-geostrophic model. In this paper, we have focused on the energy and potential enstrophy dissipation rate spectra that result from the Ekman term, when it is placed asymmetrically only on the potential vorticity equation for the lower layer but not for the upper layer. We have also considered two distinct formulations of the Ekman term. In the standard formulation, the Ekman term depends only on the streamfunction of the lower layer. In the extrapolated formulation, the Ekman term uses the streamfunctions of both layers to extrapolate a surface-layer streamfunction placed below the lower layer. Overall, the differences in the phenomenology between these two formulations are minor. 

In order to make headway in analyzing the resulting dissipation rate spectra, we have introduced a function $P(k)$ describing the distribution of energy at the wavenumber $k$ between baroclinic and barotropic energy. We have also introduced a function $\gC (k)$ describing the distribution of potential enstrophy between the upper layer and the lower layer. This makes it possible to calculate the energy dissipation rate spectrum $D_E (k)$ and the potential enstrophy dissipation rate spectrum $D_{G_2} (k)$, corresponding to the Ekman term, in terms of the total energy spectrum $E(k)$ and the functions $P(k)$ and $\gC (k)$. Our main results have been presented in the introduction of the paper and explained in detail in the body of the paper. We provide a brief summary and some concluding thoughts in the following:

First, we have shown that the functions $P(k)$ and $\gC (k)$ are restricted by an inequality that is a rigorous mathematical constraint. As a result, for wavenumbers $k\ll k_R$, when most of the energy is initially baroclinic, the potential enstrophy partition between the upper and lower layer has to be close to symmetric. As we move to wavenumbers $k$ that approach the Rossby wavenumbers $k_R$, or as more energy is converted from baroclinic to barotropic, this restriction on the distribution of potential enstrophy between the two layers is relaxed. For wavenumbers $k\gg k_R$ there is no restriction. 

Second, we have shown that for wavenumbers $k\ll k_R$, the tendency of the Ekman term is to stabilize the distribution of potential enstrophy between the two layers towards a stable fixed point distribution in which the Ekman term does not dissipate potential enstrophy. Away from the fixed point distribution, the Ekman term will either remove or inject potential enstrophy into the lower layer with a tendency to push the potential enstrophy distribution back towards the fixed point. The actual location of the fixed point is close to an equipartition of  potential enstrophy between layers and is such that more potential enstrophy is concentrated in the upper layer than the lower layer. This phenomenology is expected both for the standard and for the extrapolated formulation of the Ekman term. 

Thirdly, we have shown that in the limit $k\gg k_R$, the Ekman term, under the standard formulation, will dissipate potential enstrophy from the lower layer unconditionally. Under the extrapolated formulation, it will remain dissipative for almost all distributions of potential enstrophy between the two layers,  with the caveat that there may be a value for $\gC (k)$ very close to $1$,  corresponding to almost all potential enstrophy concentrated in the upper layer,   where the potential enstrophy dissipation rate spectrum $D_{G_2}(k)$ becomes zero and then crosses over to injecting potential enstrophy in the lower layer. If such a value of $\gC (k)$ exists, it will also be a stable fixed-point.

Finally, we have shown that, in its standard formulation, the Ekman term unconditionally dissipates energy over all wavenumbers $k$. However, in the case of the extrapolated formulation, there exist negative regions, both in the limit $k\ll k_R$ and $k\gg k_R$ where the Ekman term may be injecting energy. These negative regions are displayed in Fig.~\ref{fig:BE-one} and Fig.~\ref{fig:BE-two}. The negative region corresponding to the limit $k\ll k_R$ is inaccessible, when at least half of the energy is barotropic, as is expected in most realistic situations  at steady state, so the asymmetric Ekman term plays the same role, in the predicted Salmon phenomenology \cite{article:Salmon:1978,article:Salmon:1980,book:Salmon:1998}, of dissipating an inverse barotropic energy cascade, in both the standard and the extrapolated formulation. Furthermore, we expect that being within the negative region, either prior to reaching steady state, or in the case of very large $\nu_E$ whereby the inverse barotropic energy cascade is entirely arrested,  will tend to push the energy distribution back towards the barotropic direction. The negative region in the limit $k\gg k_R$ is accessible only when at least half of the energy is barotropic,  and being within the negative region will tend to make the energy distribution more barotropic,  pushing the dynamic further inwards the negative region.  The negative region, however, is inaccessible when at least half of the energy is baroclinic. Although there is a very small region where we can rigorously show that the Ekman term becomes injective, the true extent of the region where the Ekman term becomes injective is unclear, but, we expect it to be limited. 

For the reader that has not yet carefully followed the mathematical study of the energy and potential enstrophy dissipation rate spectra of the Ekman term, the notion that it may inject energy and potential enstrophy  may seem peculiar and unexpected, however it is a direct consequence of placing the Ekman term at the lower layer but not the upper layer. We have previously shown  \cite{article:Gkioulekas:p16} that if the same dissipation operator is applied to the streamfunction of each layer to construct the dissipation term for that layer, then the overall energy and potential enstrophy dissipation rate spectrum will be unconditionally positive over all wavenumbers. This means that if the same Ekman term is placed on both layers, using the streamfunction of the corresponding layer, then the terms will unconditionally dissipate both energy and potential enstrophy.

This study is only one first step towards understanding the role played by the Ekman term in the phenomenology of the two-layer quasi-geostrophic model. We hope that this paper will rekindle new interest in this interesting problem.

\section*{Acknowledgements}

The paper was substantially improved thanks to constructive and very helpful criticism by an anonymous referee. The author also acknowledges email correspondence with Dr. Ka-Kit Tung. The open source computer algebra software Maxima \cite{maxima} was used to conduct the very tedious calculation of the energy and potential enstrophy dissipation rate spectra. This research has not been supported by external funding agencies. 

\appendix

\section{Extrapolating the streamfunction at the surface layer}
\label{sec:app-Ekman-term}

Let $\gy_1$ be the streamfunction at the upper layer and $\gy_2$ the streamfunction at the lower layer. Likewise, let $p_1$ be the pressure corresponding to the upper layer and $p_2$ the pressure corresponding to the lower layer. The surface layer is placed below the lower layer at pressure $p_s$. Then we extrapolate the streamfunction $\gy_s$ at the surface layer by requiring that the points $(p_s, \gy_s), (p_1, \gy_1), (p_2, \gy_2)$ be collinear. It follows that:
\begin{align*}
\gy_s &= \frac{\gy_2-\gy_1}{p_2-p_1}(p_s-p_2)+\gy_2 
= \frac{p_s-p_1}{p_2-p_1}\gy_2 + \frac{p_2-p_s}{p_2-p_1}\gy_1
= \frac{p_s-p_1}{p_2-p_1}\gy_2 + \frac{p_2-p_s}{p_s-p_1}\frac{p_s-p_1}{p_2-p_1}\gy_1.
\end{align*}
We define
\begin{equation}
\gl = \frac{p_s-p_1}{p_2-p_1} \quad\text{ and } \quad
\mu = \frac{p_2-p_s}{p_s-p_1},
\end{equation}
and therefore
\begin{equation}
\gy_s = \gl\gy_2+\gl\mu\gy_1.
\end{equation}
We also note that $\gl$ and $\mu$ are closely related because
\begin{align}
\gl (\mu+1) = \frac{p_s-p_1}{p_2-p_1}\left[\frac{p_2-p_s}{p_s-p_1} + 1\right]
= \frac{p_s-p_1}{p_2-p_1}\frac{p_2-p_1}{p_s-p_1}=1,
\end{align}
which implies that $\gl = 1/(\mu+1)$, and consequently
\begin{equation}
\gy_s = \frac{\mu\gy_1+\gy_2}{\mu+1}.
\end{equation}
Now, let us consider the restrictions applicable to the parameter $\mu$. First, note that the standard Ekman term corresponds to $p_s=p_2$, which immediately gives $\mu=0$. Consequently, in the dissipation rate spectra $D_E (k)$ and $D_{G_2} (k)$, all $\mu$-dependent terms are corrections that result from splitting the surface layer away from the lower layer. More generally, we show the following propositions:

\begin{proposition}
$0<p_1<p_2<p_s \implies -1<\mu <0$.
\label{prop:mu-restrictions-one}
\end{proposition}

\begin{proof}
Let us assume that $0<p_1<p_2<p_s $. Then
\begin{equation}
\mu = \frac{p_2-p_s}{p_s-p_1} > \frac{p_2-p_s}{p_s-p_2} = -1.
\end{equation}
To justify the inequality, we note that $p_2-p_s < 0$ and that
\begin{align}
p_1 < p_2 &\implies p_s-p_1 > p_s-p_2 > 0
\implies \frac{1}{p_s-p_1} < \frac{1}{p_s-p_2}
\implies \frac{p_2-p_s}{p_s-p_1} > \frac{p_2-p_s}{p_s-p_2}.
\end{align}
We also have
\begin{equation}
\mu = \frac{p_2-p_s}{p_s-p_1} < \frac{p_s-p_s}{p_s-p_1} = 0,
\end{equation}
where, the inequality is justified by $p_2 < p_s$ and $p_s-p_1 > 0$. We conclude that $-1 < \mu < 0$.
\end{proof}

\begin{proposition}
$\ds\systwo{p_1 = 1/4 \land p_2 = 3/4}{p_2 < p_s < 1} \implies -1/3 < \mu < 0$.
\label{prop:mu-restrictions-two}
\end{proposition}

\begin{proof}
Let us assume that $p_1 = 1/4$ and $p_2 = 3/4$ and $p_2 < p_s < 1$. Then $\mu < 0$ is an immediate consequence of Proposition~\ref{prop:mu-restrictions-one}. Furthermore,
\begin{align}
\mu &= \frac{p_2-p_s}{p_s-p_1} 
= \frac{3/4-p_s}{p_s-1/4} 
= \frac{3-4p_s}{4p_s-1} 
= \frac{-(4p_s-1)+2}{4p_s-1}
= -1+\frac{2}{4p_s-1}
> -1+\frac{2}{4-1} = -\frac{1}{3}.
\end{align}
The inequality is justified by the assumption $p_s < 1$. We conclude that $-1/3 < \mu < 0$. 
\end{proof}

\section{The sign of the coefficient $B_G^{(2)}(k)$}
\label{sec:sign-of-BG-two}

In this section we derive Proposition~\ref{prop:sign-of-BG-two-prop-one}, Proposition~\ref{prop:sign-of-BG-two-prop-two}, and Proposition~\ref{prop:sign-of-BG-two-prop-three} that determine the sign of the coefficient $B_G^{(2)}(k)$. We show that for $\mu = -1/3$, $B_G^{(2)}(k)$ is unconditionally positive (see Proposition~\ref{prop:sign-of-BG-two-prop-one}), whereas for $\mu\in (-1/3, 0]$, $B_G^{(2)}(k)$ remains positive under the condition $|2\gC (k)-1| \leq\min\{1, [1-P(k)]/P(k)\}$ (see Proposition~\ref{prop:sign-of-BG-two-prop-two}). This condition is expected to hold in the limit $k\ll k_R$, but not when $k\gg k_R$, raising the question of whether the coefficient $B_G^{(2)}(k)$ can be negative in the latter limit. Proposition~\ref{prop:sign-of-BG-two-prop-three} shows that this is possible if and only if $1-P(k)<(1+3\mu)/(2+4\mu)$. 

In both Proposition~\ref{prop:sign-of-BG-two-prop-one} and Proposition~\ref{prop:sign-of-BG-two-prop-two}, we exclude the case $\gC (k)=1$, corresponding to no potential enstrophy in the lower layer, because then the potential enstrophy dissipation rate spectrum of the lower layer will be trivially zero for all wavenumbers, which is to be expected. The case $\mu = -1/3$ is handled  by Proposition~\ref{prop:sign-of-BG-two-prop-one} and the case $\mu\in (-1/3, 0]$ is handled by Proposition~\ref{prop:sign-of-BG-two-prop-two} and Proposition~\ref{prop:sign-of-BG-two-prop-three}.

 For the purpose of setting up the following arguments, let $k\in (0, +\infty)$ be a given wavenumber, and write $\gC (k) = 1/2+x$ and $P(k) = 1-y$ with $x\in [-1/2, 1/2)$ and $y\in [0, 1]$. Note that we are excluding the case $x=1/2$, corresponding to no potential enstrophy in the lower layer. Then, it follows that $B_G^{(2)}(k)$ can be rewritten as
\begin{align}
B_G^{(2)}(k) &= [-4\gC(k) P(k) + 2P(k) - 2\gC(k) + 3]+\mu [3-2\gC(k) -4P(k)] \\
&= 4xy-(2\mu+6)x+4\mu y+(2-2\mu), \label{eq:sign-of-BG-two-eq-one-xy}
\end{align}
and we use this expression as the starting point for the proofs given in the following:
\begin{proposition}
For $\mu = -1/3$ and $0 \leq \gC(k)<1$, we have $B_G^{(2)}(k)>0$ for all wavenumbers $k\in (0, +\infty)$.
\label{prop:sign-of-BG-two-prop-one}
\end{proposition}

\begin{proof}
We substitute $\mu = -1/3$ to Eq.~\eqref{eq:sign-of-BG-two-eq-one-xy} and obtain:
\begin{align}
B_G^{(2)}(k) &= 4xy-(2\mu+6)x+4\mu y+(2-2\mu) \\
&= 4xy-(16/3)x-(4/3)y+8/3 \\
&= (4/3)(3xy-4x-y+2) \\
&= (4/3)[(3x-1)y-2(2x-1)].
\end{align}
Under the assumption $0\leq \gC (k)<1$, we have $x\in [-1/2, 1/2)$, so we distinguish between the following cases: 

\noindent
\emph{Case 1:} Assume that $x\in (1/3, 1/2)$. Then, we have $2x-1<0$ and $3x-1>0$, and we also note that $y\geq 0$. It follows that
\begin{align}
B_G^{(2)}(k) &= (4/3)[(3x-1)y-2(2x-1)] \\
&\geq (4/3)[-2(2x-1)] \\ 
&= -(8/3)(2x-1)>0.
\end{align}
The first weak inequality uses $3x-1>0$ and $y\geq 0$. The second strict inequality uses $2x-1<0$. 

\noindent
\emph{Case 2:} Assume that $x\in [-1/2, 1/3]$. Then, we have $3x-1 \leq 0$, and noting also that $y\leq 1$, it follows that
\begin{align}
B_G^{(2)}(k) &= (4/3)[(3x-1)y-2(2x-1)] \\
&\geq (4/3)[(3x-1)-2(2x-1)] \\
&= (4/3)(1-x)
\geq (4/3)(1-1/3) > 0.
\end{align}
The first weak inequality uses $3x-1 \leq 0$ and $y\leq 1$. The second weak inequality uses the hypothesis $x\leq 1/3$. The subsequent strong inequality is trivial. 

In both cases, we conclude that $B_G^{(2)}(k) > 0$ for all $(x, y)\in [-1/2, 1/2)\times [0,1]$, and that concludes the argument
\end{proof}

Unfortunately, this result does not generalize for $\mu\in (-1/3, 0]$, as we can see from the counterexample corresponding to $x=1/2$ and $y=0$. Substituting these values gives $B_G^{(2)}(k) = -1-3\mu$ and we note that for $\mu>-1/3$, we have $B_G^{(2)}(k) = -1-3\mu < -1-3(-1/3) = 0$. Although, technically, we are excluding the case $x = 1/2$, because it corresponds to having no potential enstrophy in the lower layer, we expect to have $B_G^{(2)}(k) < 0$ for a range of values of $x$ near $1/2$. That said, with an additional assumption, the following proposition generalizes Proposition~\ref{prop:sign-of-BG-two-prop-one}:
\begin{proposition}
Assume that $0 \leq \gC (k) <1$. Then
\begin{equation}
\systwo{-1/3 < \mu < 0}{|2\gC (k)-1| \leq \min\left\{1, \ds\frac{1-P(k)}{P(k)}\right\}} 
\implies B_G^{(2)}(k) > 0,
\end{equation}
and
\begin{equation}
\systwo{\mu = 0}{|2\gC (k)-1| < \min\left\{1, \ds\frac{1-P(k)}{P(k)}\right\}} 
\implies B_G^{(2)}(k) > 0.
\end{equation}
\label{prop:sign-of-BG-two-prop-two}
\end{proposition}

\begin{proof}
We begin our argument by rewriting $B_G^{(2)}(k)$ as follows:
\begin{align}
B_G^{(2)}(k) &= 4xy-(2\mu +6)x+4\mu y+2-2\mu \\ 
&= 2x[2y-(\mu+3)]+4\mu y+2-2\mu.
\end{align}
From the assumptions $y\in [0,1]$ and $\mu\in (-1/3, 0)$, we observe that the expression in the bracket satisfies:
\begin{align}
2y-(\mu+3) &\leq 2-(\mu +3) = -1-\mu \\ 
&< -1+1/3 = -2/3 < 0,
\end{align}
where we have used $y\leq 1$ for the first inequality and $\mu > -1/3$ for the second inequality. 

For now, let us assume that $|2\gC (k)-1| \leq \min\{1, [1-P(k)]/P(k)\}$, which, in terms of $x$ and $y$ can be rewritten as $2x \leq \min\{1, y/(1-y)\}$. The strong inequality version of this hypothesis becomes necessary when $\mu=0$ (Case 2a, in the following), however we will invoke the stronger hypothesis only when we deal with that particular case. At this point, in order to continue our argument, it becomes necessary to distinguish between the following cases:

\emph{Case 1:} Assume that $y\in (1/2, 1]$. It follows that $y/(1-y)>1$ and therefore
\begin{equation}
2x \leq \min\left\{1, \frac{y}{1-y}\right\} = 1.
\end{equation}
This allows us to bound $B_G^{(2)}(k)$ from below, as follows:
\begin{align}
B_G^{(2)}(k) &= 2x [2y-(\mu+3)]+4\mu y+2-2\mu \\
&\geq 2y-(\mu+3)+4\mu y+2-2\mu \\
&= (2+4\mu)y-3\mu-1 \\
&> (2+4\mu)(1/2)-3\mu -1 \\
&= 1+2\mu -3\mu -1 = -\mu \geq 0,
\end{align}
which, in turn, implies that $B_G^{(2)}(k)>0$, since at least one inequality is a strict inequality. For the first inequality we have used $2x \leq 1$ and $2y-(\mu+3) < 0$. For the second strict inequality, we have used the hypothesis $y>1/2$ and also the observation that $\mu > -1/3$ implies that $2+4\mu > 2+4(-1/3) = 2/3 > 0$. The last inequality is the proposition hypothesis $\mu\leq 0$. 

\noindent
\emph{Case 2:} Assume that $y\in [0, 1/2]$. It follows that $y/(1-y) \leq 1$ and therefore Eq. (1) reduces to 
\begin{equation}
2x \leq \min\left\{1, \frac{y}{1-y}\right\} = \frac{y}{1-y}.
\end{equation}
This allows us to bound $B_G^{(2)}(k)$ from below, as shown in the following:
\begin{align}
B_G^{(2)}(k) &= 2x [2y-(\mu+3)]+4\mu y+2-2\mu \\
&\geq \frac{y[2y-(\mu+3)]}{1-y}+4\mu y + 2-2\mu \label{eq:BG2-kek} \\
&= \frac{(2-4\mu)y^2+(5\mu-5)y+(2-2\mu)}{1-y} \\
&= \frac{\phi (y,\mu)}{1-y},
\end{align}
with $\phi (y,\mu)$ given by
\begin{equation}
\phi (y,\mu) = (2-4\mu)y^2+(5\mu-5)y+(2-2\mu).
\end{equation}
The inequality in Eq.~\eqref{eq:BG2-kek} is justified by $2x\leq y/(1-y)$ and $2y-(\mu+3)<0$. Since $1-y>0$, it is sufficient to establish that $\phi (y,\mu)$ is positive for $y\in [0, 1/2]$ and $\mu\in (-1/3, 0]$. We observe that the discriminant of $\phi (y,\mu)$ with respect to $y$ is given by
\begin{align}
\gD (\mu) &= (5\mu -5)^2-4(2-4\mu)(2-2\mu) \\
&= (1-\mu)(7\mu +9)>0,
\end{align}
because $1-\mu >1>0$ and $7\mu +9 > 7(-1/3)+9 > 0$. This means that $\phi (y,\mu)$ has two zeroes $y_1$ and $y_2$ that can be calculated via the quadratic formula. We will now claim that $y_1 > y_2 \geq 1/2$. The corresponding necessary and sufficient condition, given the a priori existence of the two zeroes $y_1, y_2$, is that $\phi (1/2, \mu)$ and the coefficient of $y^2$ must have the same sign and the point $1/2$ must be on the left side of the vertex of $\phi (y,\mu)$ with respect to $y$. For the first condition, we note that
\begin{align}
(2-4\mu)\phi (1/2, \mu) &= (2-4\mu)[(2-4\mu)(1/2)^2+(5\mu-5)(1/2)+(2-2\mu)] \\
&= (2-4\mu)(-\mu/2) 
= -\mu (1-2\mu).
\end{align}
Now, let us consider separately the following subcases:

\noindent
\emph{Case 2a:} Assume that $\mu <0$. Then, it immediately follows that $(2-4\mu)\phi (1/2, \mu)>0$, since $\mu <0$ and $1-2\mu >0$, consequently $\phi (1/2, \mu)$ and the coefficient $(2-4\mu)$ of $y^2$ have the same sign. To show that $1/2$ appears to the left of the vertex of $\phi (y,\mu)$ with respect to $y$, we note that
\begin{equation}
\frac{1}{2}+\frac{5\mu-5}{2(2-4\mu)} = \frac{-3+\mu}{2(2-4\mu)} < 0,
\end{equation}
where the inequality is justified by $-3+\mu < -3 <0$ and $2-4\mu >2>0$. It follows that in this case we have $y_1 > y_2 > 1/2$ and that implies that for all $y\in [0, 1/2]$, we have $\phi (y,\mu)>0$. We conclude that
\begin{equation}
B_G^{(2)}(k) \geq \frac{\phi (y,\mu)}{1-y} > 0,
\end{equation}
and therefore $B_G^{(2)}(k) > 0$.

\noindent
\emph{Case 2b:} Assume that $\mu =0$. Then it follows that $\phi (y,0) = 2y^2-5y+2 = (y-2)(2y-1)$ which satisfies $\phi (y,\mu) >0$ for all $y\in [0, 1/2)$. However, we see that for $y=1/2$, we have $\phi (1/2, 0) = 0$. So, generally, we can claim $\phi (y,\mu)\geq 0$ for all $y\in [0, 1/2]$. Using the more powerful assumption $2x<y/(1-y)$, we can argue that
\begin{equation}
B_G^{(2)}(k) > \frac{\phi (y,\mu)}{1-y} \geq 0,
\end{equation}
and therefore $B_G^{(2)}(k) > 0$. 

We conclude that in all of the above cases, under the stated assumptions, we have $B_G^{(2)}(k) > 0$. 
\end{proof}

\begin{proposition}
There exists a $\gc(k)<1/2$ such that 
\begin{equation}
2\gC (k)-1 > \gc (k) \implies B_G^{(2)}(k)<0,
\end{equation}
if and only if $P(k)$ satisfies the inequality
\begin{equation}
1-P(k)<(1+3\mu)/(2+4\mu).
\end{equation}
\label{prop:sign-of-BG-two-prop-three}
\end{proposition}

\begin{proof}
We begin by rewriting the coefficient $B_G^{(2)}(k)$ as 
\begin{align}
B_G^{(2)}(k) &= 4xy-(2\mu+6)c+4\mu y+2-2\mu \\
&= x(4y-2\mu-6)+4\mu y +2-2\mu,
\end{align}
and noting that
\begin{equation}
4y-2\mu-6 \leq 4-2\mu-6 = -2(\mu+1)<0,
\end{equation}
where,  in the first inequality we have used $y\leq 1$ and in the second inequality we have used $\mu+1 \geq 2/3>0$. Solving for $x$, it follows that 
\begin{equation}
B_G^{(2)}(k)<0 \ifonlyif x>\frac{2\mu-2-4\mu y}{4y-2\mu-6}.
\end{equation}
Since $x$ is restricted in the interval $x\in [-1/2, 1/2]$, we can have $B_G^{(2)}(k)<0$ for some values of $x$,  if and only if 
\begin{align}
\frac{2\mu-2-4\mu y}{4y-2\mu-6} &< \frac{1}{2} \ifonlyif 2\mu-2-4\mu y > 2y-\mu-3 \\
&\ifonlyif (2+4\mu)y < 1+3\mu \\
&\ifonlyif y < \frac{1+3\mu}{2+4\mu} \\
&\ifonlyif 1-P(k) < \frac{1+3\mu}{2+4\mu}.
\end{align}
Here, the first equivalence is justified by the inequality $4y-2\mu-6< 0$ and the third  equivalence is justified by noting that $2+4\mu\geq 2+4(-1/3)>0$. This concludes the proof.
\end{proof}

Note that for the case of extrapolated Ekman dissipation (i.e. $\mu=0$), the condition on $P(k)$ in Proposition~\ref{prop:sign-of-BG-two-prop-three} reduces to $1-P(k)<1/2$, or equivalently $P(k)>1/2$, corresponding to an energy distribution that is at least half baroclinic. This condition becomes more restrictive as $\mu$  approaches $-1/3$, and for $\mu = -1/3$  we get $P(k)>1$, which is impossible, consistently with Proposition~\ref{prop:sign-of-BG-two-prop-one} that shows that $B_G^{(2)}(k)$ is unconditionally positive for the case $\mu = -1/3$.

\section{The sign of the coefficient $B_G^{(1)}(k)$}
\label{sec:sign-of-BG-one}

In this section, we give the proof for the following proposition, establishing a sufficient condition for having a positive coefficient $B_G^{(1)}(k)$.
\begin{proposition}
$\gC(k) < (2+\mu)/2 \implies B_G^{(1)}(k) > 0$.
\end{proposition}

\begin{proof}
Let us assume that $\gC(k) < (2+\mu)/2$. We also note that $0 \leq P(k) \leq 1$ and $\mu\in [-1/3, 0]$. It follows that
\begin{align}
B_G^{(1)}(k) &= 2[1-\gC (k)]+\mu [1-2P (k)] \\
&> 2[1-(2+\mu)/2]+\mu [1-2P(k)] \\
&= -\mu+\mu [1-2P(k)] \\
&= -2\mu P(k) \geq 0.
\end{align}
The first inequality is justified by the hypothesis $\gC (k) < (2+\mu)/2$. The second inequality is a consequence of $\mu\leq 0$ and $P(k)\geq 0$. We conclude that $B_G^{(1)}(k)>0$. 
\end{proof}
Note that for the case of standard Ekman dissipation (i.e. $\mu=0$), the proposition reduces to
\begin{equation}
\gC (k) < 1 \implies B_G^{(1)}(k) > 0,
\end{equation}
which is almost unconditional. For the case of extrapolated Ekman dissipation, with the surface layer placed at ground level (i.e. $\mu=-1/3$), the proposition reduces to
\begin{equation}
\gC (k) < 5/6 \implies B_G^{(1)}(k) > 0.
\end{equation}
It is easy then, to argue that the assumption $\gC(k) < 5/6$ is strong enough to ensure that $B_G^{(1)}(k)>0$ for all $\mu\in [-1/3, 0]$, as we have claimed in the main text, noting that the condition can be weakened when $\mu\in (-1/3, 0]$.

\section{Justification of geometry shown in Fig.~\ref{fig:BE-two}}
\label{app:fig-two-geometry}

In Fig.~\ref{fig:BE-two}, a negative region emerges because the graph of the curve $B_E^{(2)}(k)=0$ is situated below the graph defined by the equation $2x = \min\{1, y/(1-y)\}$. In this section, we prove this claim in detail. As was explained in Section~\ref{sec:energy-dissipation-rate-spectrum}, we have:
\begin{equation}
B_E^{(2)}(k) = 0 \ifonlyif x = \frac{(\mu +1)y}{2(1-y)}.
\end{equation}
It is obvious that this curve passes through the origin $(x, y) = (0, 0)$. For $x = 1/2$, solving for $y$ gives $y = 1/(\mu +2)$, so the curve passes also  through the point $(x, y) = (1/2, 1/(\mu +2))$. These are the two points where the two curves intersect, and since $1/(\mu +2)>1/2$, the second intersection point occurs after the ``pointy box'' curve becomes horizontal. The intersection points also indicate that $x$ increases when $y$ increases along the curve $x = [(\mu+1)y]/[2(1-y)]$, and since $x$ is a homographic function of $y$, we expect that $x$ increases as $y$ increases for all $y\in [0, 1/(\mu +2)]$. Since $x$ is a continuous function of $y$ along the curve $B_E^{(2)}(k)=0$ for all $y\in [0, 1/(\mu+2)]$, it is sufficient to show that one interior point of the curve $B_E^{(2)}(k)=0$, corresponding to $y\in (0, 1/(\mu+2))$, is below the corresponding point on the ``pointy box'' curve given by the equation $2x = \min\{1, y/(1-y)\}$. As a matter of fact, it is fairly simple to establish this for all $y\in (0, 1/2)$, by considering the vertical distance $\gD (y)$ between the two curves, as a function of $y$:
\begin{align}
\gD (y) &= \frac{(\mu+1)y}{2(1-y)} -\frac{1}{2}\min\left\{1, \frac{y}{1-y}\right\} \\
&= \frac{(\mu+1)y}{2(1-y)} - \frac{y}{2(1-y)} \\
&= \frac{\mu y}{2(1-y)}.
\end{align}
For $\mu\in [-1/3, 0)$ it follows immediately that $\gD (y) <0$ for all $y\in (0, 1/2)$ since $y>0$ and $1-y>0$. For $\mu =0$, we have instead $\gD (y)=0$, meaning that the two curves will coincide, eliminating the negative region. 

\bibliographystyle{elsarticle-num}
\bibliography{references}

\end{document}